\documentclass[10pt,journal,compsoc]{IEEEtran}

\UseRawInputEncoding

\usepackage{booktabs}

\ifCLASSOPTIONcompsoc
  \usepackage[nocompress]{cite}
\else
  \usepackage{cite}
\fi

\usepackage{amsfonts,amsthm}
\usepackage{amssymb}
\ifCLASSINFOpdf
   \usepackage[pdftex]{graphicx}
   \graphicspath{{../pdf/}{../jpeg/}}
   \DeclareGraphicsExtensions{.pdf,.jpeg,.png}
\else
   \usepackage[dvips]{graphicx}
   \graphicspath{{../eps/}}
   \DeclareGraphicsExtensions{.eps}
\fi

\makeatletter

\newcommand{\Rmnum}[1]{\expandafter\@slowromancap\romannumeral #1@}
\makeatother

\newcommand{\bm}[1]{\mbox{\boldmath{$#1$}}}

\def\W{{\mathbf W}}

\def\D{{\mathbf D}}

\def\C{{\mathbf C}}

\usepackage{amsmath,amssymb}
\newtheorem{myTheo}{Theorem}

\newtheorem{myDef}{Definition}
\newtheorem{myPro}{Proposition}
\newtheorem{myCor}{Corollary}

\newtheorem{lem}{Lemma}

\usepackage{enumitem}
\setenumerate[1]{itemsep=0pt,partopsep=0pt,parsep=\parskip,topsep=5pt}
\setitemize[1]{itemsep=0pt,partopsep=0pt,parsep=\parskip,topsep=5pt}
\setdescription{itemsep=0pt,partopsep=0pt,parsep=\parskip,topsep=5pt}

\usepackage{multirow}

\usepackage[linesnumbered,boxed,ruled,commentsnumbered]{algorithm2e}

\usepackage{algpseudocode}
\usepackage{array}
\ifCLASSOPTIONcompsoc
  \usepackage[caption=false,font=footnotesize,labelfont=sf,textfont=sf]{subfig}
\else
  \usepackage[caption=false,font=footnotesize]{subfig}
\fi
\usepackage{cite}
\usepackage{url}
\usepackage{color}
\usepackage{epstopdf}
\usepackage{overpic}

\begin{document}

\title{Deep NMF Topic Modeling}

\author{Jian-Yu Wang and Xiao-Lei Zhang
        \IEEEcompsocitemizethanks{\IEEEcompsocthanksitem
        Both authors are with the Center for Intelligent Acoustics and Immersive Communications, School of Marine Science and Technology, Northwestern Polytechnical University, Xi'an, China and the Research \& Development Institute of Northwestern Polytechnical University in Shenzhen, Shenzhen, China. E-mail: alexwang96@mail.nwpu.edu.cn, xiaolei.zhang@nwpu.edu.cn.}
\thanks{Manuscript received April 19, 2005; revised August 26, 2015.}
}

\markboth{Journal of \LaTeX\ Class Files,~Vol.~14, No.~8, August~2015}%
{Shell \MakeLowercase{\textit{et al.}}: Bare Demo of IEEEtran.cls for Computer Society Journals}

\IEEEtitleabstractindextext{%
\begin{abstract}

Nonnegative matrix factorization (NMF) based topic modeling methods do not rely on model- or data-assumptions much. However, they are usually formulated as difficult optimization problems, which may suffer from bad local minima and high computational complexity. In this paper, we propose a deep NMF (DNMF) topic modeling framework to alleviate the aforementioned problems. It first applies an unsupervised deep learning method to learn latent hierarchical structures of documents, under the assumption that if we could learn a good representation of documents by, e.g. a deep model, then the topic word discovery problem can be boosted. Then, it takes the output of the deep model to constrain a topic-document distribution for the discovery of the discriminant topic words, which not only improves the efficacy but also reduces the computational complexity over conventional unsupervised NMF methods. We constrain the topic-document distribution in three ways, which takes the advantages of the three major sub-categories of NMF---basic NMF, structured NMF, and constrained NMF respectively. To overcome the weaknesses of deep neural networks in unsupervised topic modeling, we adopt a non-neural-network deep model---multilayer bootstrap network.
To our knowledge, this is the first time that a deep NMF model is used for unsupervised topic modeling. We have compared the proposed method with a number of representative references covering major branches of topic modeling on a variety of real-world text corpora. Experimental results illustrate the effectiveness of the proposed method under various evaluation metrics.

\end{abstract}

}

\maketitle

\IEEEdisplaynontitleabstractindextext

\IEEEpeerreviewmaketitle

\IEEEraisesectionheading{\section{Introduction}\label{sec:introduction}}

\IEEEPARstart{T}{opic} modeling extracts salient features and discovers structural information from a large collection of documents \cite{blei2003latent}.
%{\color{red}{
This paper focuses on discussing the nonnative matrix factorization (NMF) based topic modeling \cite{choo2013utopian,gillis2013fast,kumar2013fast,gillis2014successive,fu2015self,8666058,li2016graph}. NMF topic modeling usually decomposes the document-word representation of documents into a topic-document matrix and a word-topic matrix. Existing decomposition methods usually have the following two major problems. First, it is challenging to discover common patterns or topics in the documents and organize them into hierarchy \cite{NIPS2003_2466,chien2015hierarchical}. Second, the topic-word distribution do not meet human interpretation of documents \cite{ramage2011partially,doshi2017towards}. For example, traditional topic modeling may lose smaller subject codes, i.e. sub-topics, in the tails of large topics, which leads to the inability of describing topic dimensions in terms of the human interpretable objects of topics, and simultaneously loses all latent sub-structure within each topic \cite{ramage2011partially}.
Deep learning, which learns hierarchical data representations, provide one solution to the aforementioned problems. However, existing deep learning methods for topic modeling are mostly supervised, and fall into the category of probabilistic topic models \cite{7258387,7869412}. To our knowledge, unsupervised deep NMF topic modeling seems unexplored yet, due to maybe the high computational complexity of deep unsupervised NMF \cite{trigeorgis2016deep,ZONG201774} as well as the lack of supervised information of data.

\subsection{Contributions}

In this paper, we aim to explore an unsupervised deep NMF (DNMF) framework to address the above challenges. Because modeling topic hierarchies of documents and discovering topic words simultaneously is a complicated optimization problem, we propose to solve the two problems in sequence, under the assumption that, if the representation of documents is good enough, then the overall performance can be boosted \cite{xie2013integrating}. The proposed method contains the following novelties:

\begin{itemize}
  \item An unsupervised deep NMF framework is proposed. It first learns the topic hierarchies of documents by an unsupervised deep model, whose output is used to constrain the topic-document matrix. Then, it produces a good solution to the topic-document matrix and word-topic matrix by NMF under the constraint.
      It can have many implementations by incorporating different NMF methods and deep models. Unlike conventional NMF topic modeling methods that make predefined assumptions, DNMF alleviates the weaknesses of NMF, e.g. non-unique factorization, by deep learning. To our knowledge, this is the first work of unsupervised deep NMF for topic modeling.

  \item Three implementations of DNMF that reach the state-of-the-art performance are proposed.
      The three algorithms fall into the three major subclasses of NMF technologies \cite{wang2012nonnegative}, denoted as basic DNMF (bDNMF), strutured DNMF (sDNMF), and constrained DNMF (cDNMF) respectively. Specifically, bDNMF takes the output of the deep model as the topic-document matrix directly to generate the word-topic distribution. sDNMF takes the output of the deep model as the intrinsic geometry of the topic-document distribution, which is used to mask the topic-document matrix. cDNMF takes the output of the deep model as a regularization of the topic-document distribution. The convergence of the proposed algorithms is theoretically proved.

  \item Because the representation of documents in topic modeling is usually sparse and high-dimensional, existing deep neural networks can easily overfit to the documents. Although some methods reduce the dimension of the documents by discarding low-frequency words, their performance suffers from the compromise \cite{xie2016unsupervised}. To address the problem, this paper applies multilayer bootstrap networks (MBN) to learn the topic hierarchies of documents. MBN contains three simple operators---random resampling, stacking, and one-nearest-neighbor optimization. To our knowledge, this is the first time that a non-neural-network unsupervised deep model is applied to topic modeling, which outperforms conventional shallow topic modeling methods significantly.
\end{itemize}
%We have compared the XXX. Experimental results show that...
We have compared the proposed DNMF variants with 9 representative topic modeling methods  \cite{papadimitriou2000latent,blei2003latent,cai2008modeling, cai2009probabilistic,gillis2013fast,gillis2014successive,kumar2013fast ,fu2018anchor,henao2015deep} covering probabilistic topic models \cite{papadimitriou2000latent,blei2003latent,cai2008modeling
,cai2009probabilistic}, NMF methods  \cite{gillis2013fast,gillis2014successive,kumar2013fast,fu2018anchor}, and deep topic models  \cite{henao2015deep}.
Empirical results on the 20-newsgroups, topic detection and tracking database version 2 (TDT2), and Reuters-21578 corpora illustrate the effectiveness of DNMF in terms of three evaluation metrics. Moreover, the hyperparameters of the DNMF variants have stable working ranges across all situations, which facilitates their practical use.

In this paper, we first introduce some related work and preliminaries in the following two subsections, then present the proposed DNMF framework and its three implementations in Section \ref{Algorithm}. Section \ref{Exp} presents the experimental results. Finally, Section \ref{Conclusion} concludes our findings.

\subsection{Related work}
\label{Background}

\textbf{\textit{Probabilistic topic modeling}}: Topic models were originally formulated as unsupervised probabilistic models \cite{blei2003latent,teh2005sharing,li2017supervised,cai2008modeling}. A seminal work of probabilistic topic models is latent Dirichlet allocation (LDA) \cite{blei2003latent}.
It models a document as a multinomial distribution over latent semantic topics, and models a topic itself as a multinomial distribution over words.
The document-dependent topic embedding, governed by a Dirichlet prior, is estimated in an unsupervised way and then adopted as the low-dimensional feature for document classification and indexing.
Later on, hierarchical tree-structured priors such as nested Dirichlet processing \cite{teh2005sharing,paisley2014nested} or nested Chinese restaurant process \cite{blei2010nested,paisley2014nested} were applied to discover the hierarchy of topics and capture the nonlinearity of documents.
However, the hierarchical probabilistic models suffer from conceptual and practical problems.
For example, their optimization problem is NP-hard in the worst case due to the intractability of the posterior inference \cite{arora2012learning}.
Existing methods have to resort to approximate inference methods, such as variational Bayes and Gibbs sampling which is also difficult to carry out \cite{chien2017deep}.
Besides, because the exact inference is intractable, the models can never make predictions for words that are sharper than the distributions predicted by any of the individual topics.
As a result, the hypothesis of probability distributions are unable to be applied to all text corpora \cite{hinton2009replicated}. Moreover, there is a lack of justification of the Bayesian priors as well \cite{gerlach2018network}.
Recently, a geometric Dirichlet means algorithm \cite{yurochkin2016geometric}, which builds upon a weighted $k$-means clustering procedure and is augmented with a geometric correction, overcomes the computational and statistical inefficiencies encountered by probabilistic topic models based on Gibbs sampling and variational inference.
However, the learned topic polytope is largely influenced by the performance of the clustering algorithm.

\textbf{\textit{Deep probabilistic topic modeling}}:
Another solution to the optimization difficulty of the hierarchical probabilistic models is to integrate the perspectives of the probabilistic models and deep neural networks. The integrated methods, named deep neural topic models, introduce neural network based priors as alternatives to Dirichlet process based priors \cite{larochelle2012neural,zheng2015deep,gan2015scalable,ranganath2015deep}. This integrates the powerfulness of neural network architecture into the inference of the probabilistic graph models, which makes the models not only interpretable but also powerful and easily extendable. However, they still fail to consider the veracity of the Bayesian hypothesis. The problem of component collapsing may also lead to bad local optima of the inference network in which all topics are identical.

\textbf{\textit{NMF topic modeling}}:
To deal with the optimization difficulty of the hierarchical probabilistic models, a large effort has been paid on polynomial time solvable topic modeling
algorithms, many of which are formulated as separable nonnegative matrix factorization (NMF) methods \cite{choo2013utopian,gillis2013fast,kumar2013fast,gillis2014successive,fu2015self,8666058,li2016graph}. They find the underlying parameters of topic models by decomposing the document-word data matrix into a weighted combination of a set of topic distributions \cite{ren2017spectral}. A key problem in the context of NMF research is the separability issue, i.e., whether the matrix factors are unique \cite{donoho2004does}. When one applies NMF to topic modeling, the separability assumption is equivalent to an anchor-word assumption which assumes that every topic has a characteristic anchor word that does not appear in the other topics \cite{gillis2013fast,kumar2013fast,gillis2014successive,fu2015self}.
However, because words and terms have multiple uses, the anchor word assumption may not always hold. How to avoid the unrealistic assumption is a key research topic. One solution explores tensor factorization models with three- or higher-order word co-occurence statistics. However, such statistics need many more samples than lower-order statistics to obtain reliable estimates, and separability still relies on additional assumptions \cite{fu2018anchor}, such as consecutive words being persistently drawn from the same topic. Another recent solution is anchor-free correlated topic modeling (AnchorFree) with second-order co-occurrence statistics. However, an assumption called sufficiently scattered condition is still needed to be made, though the assumption is much milder than the anchor-word assumption. Besides the problem of making additional assumptions to the data, NMF is also formulated as a shallow learning method with no more than one nonlinear layer, which may not capture the nonlinearity of documents and the hierarchy of topics well.

\textbf{\textit{Deep NMF methods}}:
The aforementioned NMF topic models are all shallow models, which is not powerful enough to grasp the nonlinearity of documents. In the NMF research community, a lot of efforts have been paid on the multilayered NMF algorithms with applications to image processing \cite{8943941,li2018deep}, speech separation \cite{7177933,8170034}, community detection \cite{ye2018deep}, etc. The basic idea is to factorize a matrix into multiple factors, where the factorization can be either linear or nonlinear. If the factorization is nonlinear, then the method is called a deep NMF. For example, deep semi-NMF \cite{trigeorgis2016deep} factorizes the basis matrix into multiple factors with the optimization criterion of minimum reconstruction error, where it does not require the factorized weight matrix to be nonnegative anymore.
Deep nonnegative basis matrix factorization \cite{ZONG201774} conducts deep factorization to the coefficient matrix with different regularization constraints on the basis matrix. However, because the bag-of-words representation of documents is high-dimensional and sparse, the application of the aforementioned idea to topic modeling is computationally high and may also suffer from overfitting. To our knowledge, no deep NMF topic modeling methods have been proposed yet.

\textbf{\textit{Unsupervised deep learning for document clustering}}:
Document clustering and topic modeling are two closely related tasks \cite{xie2013integrating}.
Unsupervised topic modeling projects documents into a topic embedding space, which promotes the development of document clustering.
Recently, many works focused on learning the representations and topic assignments of documents simultaneously by deep neural networks \cite{xie2016unsupervised,yang2017towards,MORADIFARD2020185,jabi2019deep}.
However, current deep learning methods for document clustering do not show advances over the shallow learning methods, such as NMF-based topic modeling. We conjecture that existing methods may not be good at dealing with sparse and high-dimensional representations of documents. As a compromise, they reduce the dimension of the sparse data by discarding the low-frequency words, which may significantly lose useful information. To deal with the aforementioned problems, here we develop deep models that are able to outperform conventional shallow models without discarding the low-frequency words. Note that, although some deep learning based topic models apply word embeddings to deep topic models \cite{cao2015novel,pmlr-v80-zhao18a}, it may not be unsuitable to compare them with the conventional topic modeling methods that work with the term frequency-inverse document frequency (TF-IDF) statistics.

\subsection{Preliminaries}

\subsubsection{Notations}
We first introduce some notations here. Regular letters, e.g. $\delta$, $M$, $t$, and $0$, indicate scalars. Bold lower-case letters, e.g. $\mathbf{d}$, indicate vectors. Bold capital letters, e.g. $\mathbf{D}$, $\mathbf{C}$, and $\mathbf{W}$, indicate matrices. The bold digit $\mathbf{0}$ indicates an all-zero vector or matrix. The operator $^T$ denotes the transpose. The notation $[\mathbf{C}]_{ij}$ indicates the elements of the matrix $\mathbf{C}$ at the $i$th column and $j$th row.
The operator $\odot$ is the Hadamard multiplication. The operator $\mbox{Tr}(\cdot)$ denotes the trace of matrices.

\subsubsection{Background}
In topic modeling, given a corpus of $N$ documents with $K$ topics and a vocabulary of $V$ words, denoted as $\{\mathbf{d}_{n}\}_{n=1}^N$ where $\mathbf{d}_n = [d_{n,1},\ldots,d_{n,V}]^T$ with $d_{n,v}$ as the frequency of the $v$th word in the vocabulary that appears in the $n$th document.
we aim to learn a topic-document matrix $\mathbf{W}=[w_{k,n}]\in \mathbb{R}_{+}^{K \times N}$ and a word-topic matrix $\mathbf{C}=[c_{v,k}]\in \mathbb{R}_{+}^{V \times K}$ from the document-word matrix $\mathbf{D}=[\mathbf{d}_{1},\ldots,\mathbf{d}_N]\in \mathbb{R}_{+}^{V \times N}$, where the notation $k\leq K$ is the topic index, $w_{k,n}$ is the topic label which describes the probability of the $n$th document belonging to the $k$th topic, and $c_{v,k}$ is the probability of the $v$th vocabulary that appears in the $k$th topic.
The task of topic modeling is to find an approximate factorization:
\begin{equation}\label{eq_phydic}
\mathbf{D} \approx \mathbf{C}\mathbf{W}
\end{equation}

NMF measures the distance between $\mathbf{D}$ and $\mathbf{C}\mathbf{W}$ by the squared Frobenius norm, and formulates the topic modeling problem as the following optimization problem:
\begin{equation}\label{eq4}
  (\mathbf{C},\mathbf{W}) = \arg \min_{\mathbf{C}\geq \mathbf{0};\mathbf{W}\geq \mathbf{0}}\parallel{\mathbf{D}-\mathbf{C}\mathbf{W}}\parallel^{2}_{F}
\end{equation}
where the nonnegative constraints make the solution interpretable.
Under the anchor-word assumption, the word distribution $\mathbf{C}$ is enforced to be a block diagonal matrix, which guarantees a consistent solution \cite{arora2012learning,arora2013practical}.
However, the anchor-word assumption is fragile in practice. Recently, many methods have been proposed to overcome this assumption \cite{fu2018anchor,DBLP:conf/aistats/JangH19}.

\section{Deep NMF topic model}
\label{Algorithm}

In this section, we first present the DNMF topic modeling framework in Section \ref{subsec:obj}, then implements three DNMF topic modeling methods named bDNMF, cDNMF, and sDNMF respectively in Section \ref{subsec:lasso}, and finally introduce the unsupervised deep model in Section \ref{subsec:mbn}.

\subsection{The framework of DNMF topic modeling}\label{subsec:obj}

Traditional NMF topic modeling aims to learn a document representation by linear NMF essentially. In order to capture the manifold structure or topic hierarchies of documents, a natural way is to extend NMF into a deep NMF framework. Here we propose a DNMF framework which constrains the topic-document matrix by an unsupervised deep model with multiple layers of nonlinear transforms:
\begin{equation}\label{eq:dtm}
  \begin{array}{c}
\mathbf{D} \approx \mathbf{C}\W \\
\mbox { subject to }  \quad g(\W| f(\mathbf{D})) \geq 0, \mbox{ }\mathbf{C}\geq 0,\mbox{ } \W\geq 0
  \end{array}
\end{equation}
where $f(\cdot)$ is the unsupervised deep model and $g(\cdot)$ is a discriminator used to constrain $\W$ by $f(\cdot)$. $f(\cdot)$ performs like a prior that constrains the solution of $\W$ and $\mathbf{C}$ to be interpretable and discriminant, which is the fundamental difference between DNMF and conventional NMF topic models.
 The framework is illustrated in Fig. \ref{fig1}. It minimizes the reconstruction error between $\mathbf{D}$ and $\mathbf{C}\mathbf{W}$ in terms of the squared Frobenius norm.

 A direct thought to solve problem \eqref{eq:dtm} is to optimize $f(\D)$, $\mathbf{W}$, and $\mathbf{C}$ alternatively until convergence. However, it is too costly to train a deep model in a single iteration. In practice, we take the following optimization algorithm to solve problem \eqref{eq:dtm}:
  \begin{itemize}
 \item Pretrain $f(\D)$ first by an unsupervised deep model.
 \item Optimize $\W$ and $\mathbf{C}$ alternatively with $f(\mathbf{D})$ fixed until convergence.
 \end{itemize}
 The effectiveness of the above algorithm relies on the assumption that, if a high-quality $f(\mathbf{D})$ is obtained as a prior, then the solution of $\mathbf{C}$ and $\W$ is also boosted.

 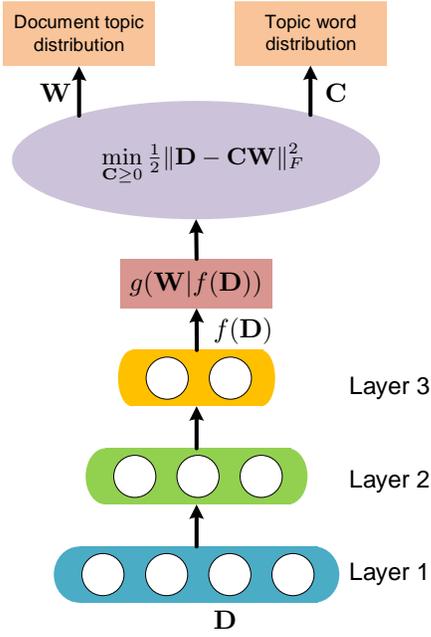
\begin{figure}[tp]
  \centering
  %\includegraphics[width = 6.5cm]{DNMF1.pdf}\\
   \begin{overpic}[width = 6.5cm]{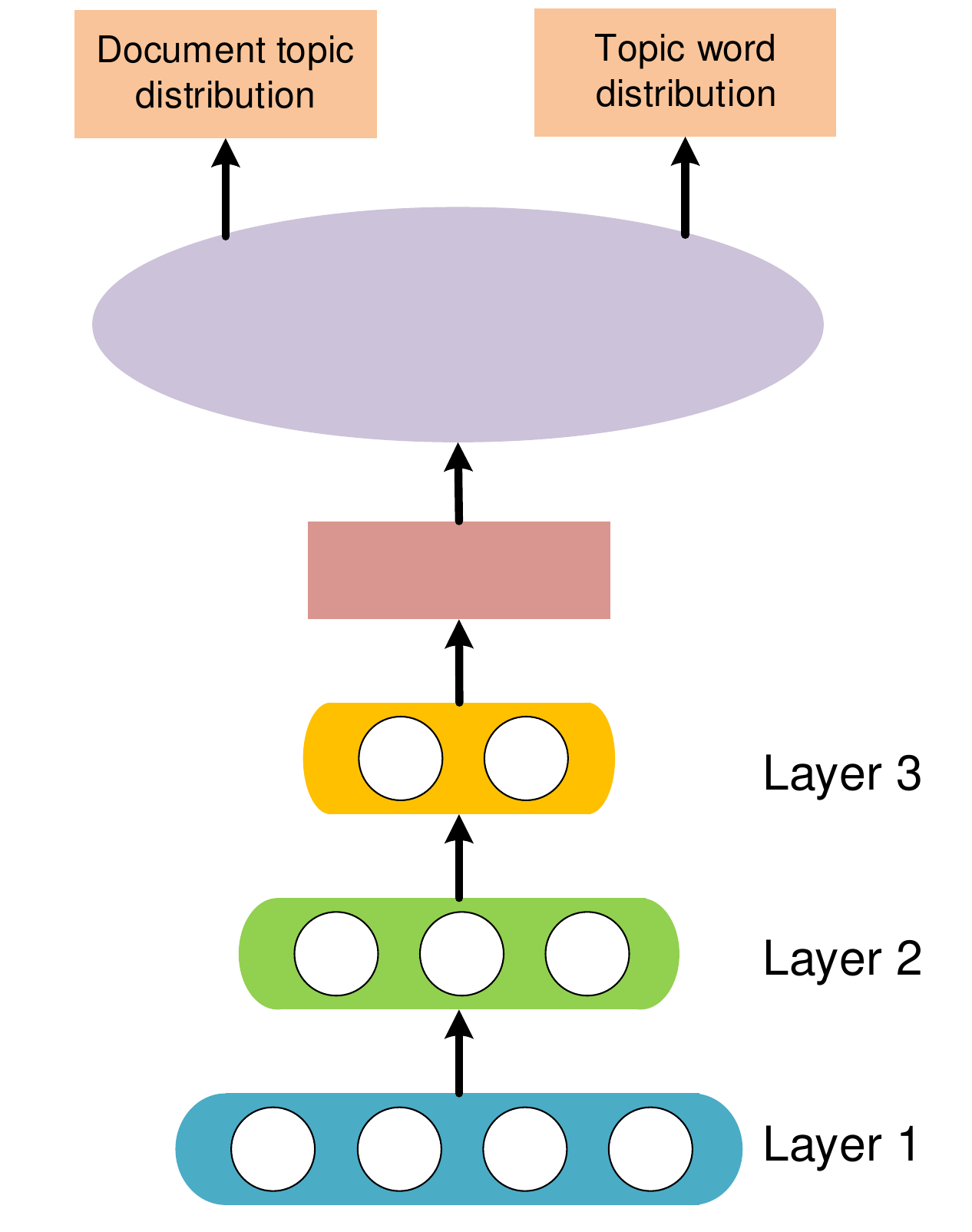}
   \put(12,83){\normalsize $\textcolor{black}{\W}$}
    \put(58,83){\normalsize $\textcolor{black}{\mathbf{C}}$}
    \put(22,73){\small $\textcolor{black}{\min\limits_{\mathbf{C} \geq 0}  \frac{1}{2}\| \mathbf{D} - \mathbf{C}\W\|_F^2}$}
    \put(26.5,52.5){\normalsize $\textcolor{black}{g(\W|f(\mathbf{D}))}$}
    \put(40,45){\normalsize $\textcolor{black}{f(\mathbf{D})}$}
%    \put(45,34){\normalsize $\textcolor{black}{f(f(\mathbf{D}^+))}$}
%    \put(45,16){\normalsize $\textcolor{black}{f(\mathbf{D}^+)}$}
   \put(40,-2){\normalsize $\textcolor{black}{\mathbf{D}}$}
   \end{overpic}
  \vspace{-0.1cm}
  \caption{The proposed DNMF framework.}\label{fig1}
 \end{figure}

The difference between the proposed topic modeling method and existing deep NMF methods \cite{trigeorgis2016deep,ZONG201774} is that the proposed method takes the deep model $f(\cdot)$ as an additional constraint of $\W$, while the methods in \cite{trigeorgis2016deep,ZONG201774} decomposes $\W$ directly into a hierarchical network. It is easy to see that the proposed framework can employ various unsupervised deep models to bring additional information into the matrix decomposition problem for specific applications.
 It is easy to constrain $\W$ flexibly as we will do in Section \ref{subsec:lasso}, which brings advanced NMF methods into the proposed framework.
 It also can either employ a pretrained deep model or conduct joint optimization of the deep model and the matrix decomposition. On the contrary, although \cite{trigeorgis2016deep,ZONG201774} can be applied to topic modeling, the computational complexity of their multilayered matrix decomposition is too high to be applied to topic modeling in practice. To our knowledge, they were not applied to topic modeling yet.

\subsection{DNMF implementations}\label{subsec:lasso}

In this subsection, we first introduce three DNMF implementations that extend the three sub-categories of NMF \cite{wang2012nonnegative} to their deep versions respectively, and then discuss the connection between the three implementations. Note that, besides the novelty of the DNMF framework, cDNMF and sDNMF are also fundamentally new even without the deep model $f(\D)$.

%\subsubsection{Geometric interpretation of DNMF topic modeling}

\subsubsection{Basic DNMF topic modeling}

\begin{algorithm}[tp]
\label{alg:bDNMF}
\small
%  \SetAlgoNoLine   % 不要算法中的竖线
  \SetKwInOut{Input}{\textbf{Input}}   % 替换关键词
  \SetKwInOut{Output}{\textbf{Output}}
  \caption{ bDNMF. }
  \Input{
         Text corpus $\mathbf{D}$, the number of topics $T$, hyperparameters $\delta \geq  0$ and  $M \geq 0$. \\
        }
  \Output{$\mathbf{C}^{(t)}$, $\W$.}
    \BlankLine
    Initialize: topic-word distribution $\mathbf{C}^0$, $ t \leftarrow 0 $;\\
    Construct a document-topic distribution $f(\mathbf{D})$ by deep unsupervised learning methods; \\
    $\W\leftarrow f(\D)$;\\
    \Repeat{\text{convergence}}
        {
          Calculate $\mathbf{C}^{(t)}$ by \eqref{eq:MU_SED}; \\
          $ t \leftarrow t+1 $;\\
          }
\end{algorithm}

Many NMF topic modeling methods introduce polytope to interpret the geometry of documents \cite{yurochkin2019scalable,fu2019nonnegative}, i.e. $[\mathbf{D}]_{ij} = \sum_{k=1}^K [\mathbf{C}]_{ik}[\mathbf{W}]_{kj}$. A standard NMF topic modeling can always find an infinite solutions of $\C$ and $\W$ that satisfy $\D \approx \C\W$. To prevent such infinite solutions, various constraints have to be added. One of the simplest constraint is to provide one of the two factors beforehand, e.g. $\W$. However, it seems not easy to find a satisfied $\W$ beforehand in history. Fortunately, deep learning provides such an opportunity. We conjecture bravely that, if a good topic-document matrix $\W$ could be learned beforehand by deep learning, then the problem of finding the other factor $\C$ can be greatly simplified, which motivates bDNMF.

Given a latent document topic proportions $f(\mathbf{D})$ from a deep model, bDNMF interprets the documents by
\begin{equation}\label{eq:Topic_model}
[\mathbf{D}]_{ij} = \sum_{k=1}^K [\mathbf{C}]_{ik}[f(\mathbf{D})]_{kj} \mbox{ for } i=1,\dots,V; m = 1,\dots,N.
\end{equation}
It is a special case of the framework in Fig. \ref{fig1} where $g(\cdot)$ is simply defined as $\mathbf{W}-f(\D) = \mathbf{0}$. Solving the factorization \eqref{eq:Topic_model} in the NMF framework results in the following optimization problem:
\begin{equation}\label{eq:SED}
\begin{split}
 \min\limits_{\mathbf{C} \geq \mathbf{0}, f(\cdot)} D_{F}[\mathbf{D}||\mathbf{C}f(\mathbf{D})]
\end{split}
\end{equation}
where ${D}_F[\mathbf{D}||\mathbf{C}f(\mathbf{D})]$ denotes the Frobenius norm of NMF with $\mathbf{C}f(\mathbf{D})$ being an approximation of $\mathbf{D}$:
\begin{equation}\label{eq:solve_SED}
 D_{{F}}[\mathbf{D}||\mathbf{C}f(\mathbf{D})]  %= \frac{1}{2}|| \mathbf{D} - \mathbf{C}f(\mathbf{D})||^2_F
 = \left\|\mathbf{D}-\mathbf{C}f(\mathbf{D})\right\|^2_{F}
\end{equation}

%As described in \cite{yurochkin2016geometric}, each topic in bDNMF is represented by a point $[\mathbf{C}]_{ik}$ lying in the $V-1$ dimensional probability simplex. Word-topic distribution can be viewed as a convex hull of the K topics $[\mathbf{C}]_{\cdot k}$, and each document corresponds to a point $[\mathbf{D}]_{\cdot j}$ lying inside the topic polytope.

%bDNMF aims to find an approximate factorization $\mathbf{D} \approx \mathbf{C}f(\mathbf{D})$.
%In words, we aim to find a convex topic polytope, which is closest to words of the observed documents.
%Since we obtained the document distributions over the topics from deep model, the difficulty is how to correct the vertices of the topic polytope.
%As Fig. \ref{fig1}, we adopted the DNMF framework to solve the above difficulties.
%}}

We solve bDNMF in two steps. First, we generate the sparse representation of documents $f(\D)$ by a deep model. Then, problem \eqref{eq:SED} is formulated as a nonnegative least squares optimization problem, which can be solved by gradient descent algorithms or multiplicative update rules \cite{lee2001algorithms}. Here we prefer multiplicative update rules, since they do not have tunable hyperparameters. As we can see, when $f(\D)$ is given, problem \eqref{eq:SED} satisfies the following first-order Karush-Kuhn-Tucker (KKT) optimality conditions:
\begin{eqnarray}
\left\{\begin{array}{ll}
\mathbf{C} \geq \mathbf{0} \vspace{1ex} \\
\frac{\partial D_F(\mathbf{D}||\mathbf{C}f(\mathbf{D}))}{\partial \mathbf{C}}\geq \mathbf{0} \vspace{1ex} \\
\mathbf{C}\odot\frac{\partial D_F(\mathbf{D}||\mathbf{C}f(\mathbf{D}))}{\partial \mathbf{C}}=\mathbf{0}
 \end{array}
 \right.\label{eq:KKT_SED}
 \end{eqnarray}
 which guarantees that the solution of \eqref{eq:SED} converges to a stationary point.

 The multiplicative update rules are described as follows. Let $\bm\Psi$ be the Lagrange multiplier of the constraint $\mathbf{C} \geq \mathbf{0}$, the Lagrangian $\mathcal{J}$ for \eqref{eq:SED} is
\begin{equation}\label{eq:Lagrange_SED}
\begin{split}
\mathcal{J} = & \mbox{Tr}(\mathbf{D}\mathbf{D}^{T}) - 2\mbox{Tr}(\mathbf{D}f(\mathbf{D})^T{\mathbf{C}}^{T}) \\
& + \mbox{Tr}(\mathbf{C}f(\mathbf{D})f(\mathbf{D})^T{\mathbf{C}}^{T}) + \mbox{Tr}(\mathbf{C}\bm\Psi)
\end{split}
\end{equation}
The partial derivative of $\mathcal{J}$ with respect to $\mathbf{C}$ is
\begin{equation}\label{eq:dif_Lagrange_SED}
\begin{split}
\frac{\partial \mathcal{J}}{\partial \mathbf{C}} = -2\mathbf{D}f(\mathbf{D})^{T} + 2\mathbf{C}f(\mathbf{D})f(\mathbf{D})^T + \bm\Psi
\end{split}
\end{equation}
By the KKT condition $\mathbf{C}\odot\bm\Psi = \mathbf{0}$,
we derive
\begin{equation}\label{eq:dif_Lagrange_SED}
2\mathbf{C}\odot(\mathbf{C}f(\mathbf{D})f(\mathbf{D})^T) - 2\mathbf{C}\odot(\mathbf{D}f(\mathbf{D})^{T}) + \mathbf{C}\odot\bm\Psi = \mathbf{0}
\end{equation}
Therefore, the multiplicative update rules for $\mathbf{C}$ can be inferred as follows:
\begin{equation}\label{eq:MU_SED}
\begin{split}
[\mathbf{C}]^{(t+1)}_{ij} \leftarrow [\mathbf{C}]^{(t)}_{ij}\frac{[\mathbf{D}f(\mathbf{D})^T]_{ij}} {[\mathbf{C}f(\mathbf{D})f(\mathbf{D})^T]_{ij}}
\end{split}
\end{equation}
where the superscript $^{(t)}$ denotes the $t$th iteration of the multiplicative update rules.

bDNMF is summarized in Algorithm \ref{alg:bDNMF}. It implements $g(\mathbf{W}, f(\mathbf{D}))$ by simply setting $\mathbf{W}= f(\D)$. The main merit of bDNMF is that it can easily get the global optimum solution of $\C$ given $\W$ fixed, which avoids the non-unique solution of the NMF topic modeling in a simple way. Its effectiveness is largely affected by $f(\D)$. In practice, we implement $f(\mathbf{D})$ as semantic topic labels, which is obtained by the deep-learning-based document clustering.

%The geometric interpretation of bDNMF is also clear.
%The word-topic matrix $\mathbf{C}$ can be represented as points on the $(V-1)$-dimensional simplex of all possible decomposition factors. Each document $\mathbf{d}$ corresponds to a point lying inside the polytope \cite{10.1145/312624.312649}.
%$f(\mathbf{D})$, which is the topic assignment of the documents, can be regarded as the weight assignment of the vertices in the topic polytope.

\subsubsection{Structured DNMF topic modeling}

Although bDNMF is simple, it reduces NMF with only one variable when $f(\D)$ is given, which limits the flexibility of $\C$. To solve the problem, sDNMF modifies the regular factorization formulation \eqref{eq4} by a new discriminator $\W = f(\D)\odot \mathbf{T}$ instead of taking $\W = f(\D)$ where $\mathbf{T}$ is a new variable. Its objective function is formulated as follows:
\begin{equation}\label{eq:cDNMF}
\begin{split}
 \min\limits_{\mathbf{C} \geq \mathbf{0},\mathbf{T}\geq\mathbf{0},  f(\cdot)} D_F(\mathbf{D}\| \mathbf{C}(f(\mathbf{D})\odot\mathbf{T})) \\
  =\min\limits_{\mathbf{C} \geq \mathbf{0},\mathbf{T}\geq\mathbf{0},  f(\cdot)} \| \mathbf{D} - \mathbf{C} (f(\mathbf{D})\odot\mathbf{T}) \|_F^2
\end{split}
\end{equation}

\begin{algorithm}[tp]
\label{alg:sDNMF}
\small
\setlength{\tabcolsep}{0.1mm}{
%  \SetAlgoNoLine   % 不要算法中的竖线
  \SetKwInOut{Input}{\textbf{Input}}   % 替换关键词
  \SetKwInOut{Output}{\textbf{Output}}
  \caption{ sDNMF. }
  \Input{
         Text corpus $\mathbf{D}$,
         the number of topics $T$,
         hyperparameters $\delta \geq  0$ and  $M \geq 0$.\\
        }
    \Output{$\mathbf{C}^{(t)}$, $f(\mathbf{D})$.}
    \BlankLine
    Initialize: topic-word distribution $\mathbf{C}^0$,
          document-topic distribution $\mathbf{T}^0$,
          weight matrix $\mathbf{T}^0$, $ t \leftarrow 0 $;\\
    Construct a document-topic distribution $f(\mathbf{D})$ by deep
unsupervised learning methods; \\
    \Repeat{\text{convergence}}
        {
          Calculate $\mathbf{T}^{(t)}$ by \eqref{eq:MU_cDNMF_W};\\
          Calculate $\mathbf{C}^{(t)}$ by \eqref{eq:MU_cDNMF_C}; \\
          $ t \leftarrow t+1 $; \\
          }
    }
\end{algorithm}

Like bDNMF, we solve sDNMF by first generating $f(\D)$ by a deep model, which formulates problem \eqref{eq:cDNMF} as an alternative least squares optimization problem. As we can see, when $f(\D)$ is given, problem \eqref{eq:cDNMF} satisfies the following first-order KKT optimality conditions:
\begin{eqnarray}
\left\{ \begin{array}{ll}
\mathbf{C} \geq \mathbf{0}, \mathbf{T} \geq \mathbf{0} \vspace{1ex} \\
\frac{\partial{D_F(\mathbf{D}\| \mathbf{C}(f(\mathbf{D})\odot\mathbf{T}))}}{\partial\mathbf{C}} \geq \mathbf{0} \vspace{1ex} \\
\mathbf{C}\odot\frac{\partial{ D_F(\mathbf{D}\| \mathbf{C}(f(\mathbf{D})\odot\mathbf{T}))}}{\partial\mathbf{C}} =\mathbf{0} \vspace{1ex} \\
\frac{\partial{D_F(\mathbf{D}\| \mathbf{C}(f(\mathbf{D})\odot\mathbf{T}))}}{\partial\mathbf{T}} \geq \mathbf{0} \vspace{1ex} \\
\mathbf{T}\odot\frac{\partial{D_F(\mathbf{D}\| \mathbf{C}(f(\mathbf{D})\odot\mathbf{T}))}}{\partial\mathbf{T}} =\mathbf{0} \vspace{1ex} \\
\end{array}\right.\label{eq:KKT_sDNMF}
\end{eqnarray}
 which guarantees that the solution of \eqref{eq:cDNMF} converges to a stationary point.

Let $\mathcal{\mathbf{U}}$ and $\mathcal{\mathbf{V}}$ denote the Lagrange multipliers of $\mathbf{C}$ and $\mathbf{T}$ respectively.
Then, minimizing \eqref{eq:cDNMF} is equivalent to minimizing the Lagrangian $\mathcal{J}$:
\begin{equation}\label{eq:jcDNMF}
\begin{split}
\mathcal{J} = D_F(\mathbf{D}\| \mathbf{C}\mathbf{T}, f(\mathbf{D})) + \mbox{Tr}(\mathcal{\mathbf{U}}\mathbf{C}^T) + \mbox{Tr}(\mathcal{\mathbf{V}}\mathbf{T}^T)
\end{split}
\end{equation}
Taking partial derivatives in \eqref{eq:jcDNMF} derives
\begin{equation}\label{eq:difcDNMF_C}
\begin{split}
\frac{\partial \mathcal{J}}{\partial \mathbf{C}} = & 2\mathbf{C}(f(\mathbf{D})\odot\mathbf{T})^T(f(\mathbf{D})\odot\mathbf{T})\\ & - 2\mathbf{D}^T(f(\mathbf{D})\odot\mathbf{T}) + \mathcal{\mathbf{U}}
\end{split}
\end{equation}
\begin{equation}\label{eq:difcDNMF_W}
\begin{split}
\frac{\partial \mathcal{J}}{\partial \mathbf{T}} = & 2((f(\mathbf{D})\odot\mathbf{T})\mathbf{C}\mathbf{C}^T)\odot f(\mathbf{D}) \\
& - 2(f(\mathbf{D})\odot\mathbf{D}\mathbf{C}^T) + \mathcal{\mathbf{V}}
\end{split}
\end{equation}
Combining with the KKT conditions, we obtain the update rules:
\begin{equation}\label{eq:MU_cDNMF_W}
\begin{split}
[\mathbf{T}]^{(t+1)}_{ij} \leftarrow [\mathbf{T}]^{(t)}_{ij}\frac{[(\mathbf{D}\mathbf{C}^T)\odot f(\mathbf{D})]_{ij}} {[((f(\mathbf{D})\odot\mathbf{T})\mathbf{C}\mathbf{C}^T)\odot f(\mathbf{D})]_{ij}}
\end{split}
\end{equation}
\begin{equation}\label{eq:MU_cDNMF_C}
\begin{split}
[\mathbf{C}]^{(t+1)}_{ij} \leftarrow [\mathbf{C}]^{(t)}_{ij}\frac{[(f(\mathbf{D})\odot\mathbf{T})^T\mathbf{D}]_{ij}} {[(f(\mathbf{D})\odot\mathbf{T})^T(f(\mathbf{D})\odot\mathbf{T})\mathbf{C}]_{ij}}
\end{split}
\end{equation}

sDNMF is summarized in Algorithm \ref{alg:sDNMF}. It promotes the effectiveness of $\C$ by introducing the internal variable $\mathbf{T}$ to bridge the gap between $f(\D)$ and $\C$.

\subsubsection{Constrained DNMF topic modeling}

\begin{algorithm}[tp]
\label{alg:DNMF}
\small
%  \SetAlgoNoLine   % 不要算法中的竖线
  \SetKwInOut{Input}{\textbf{Input}}   % 替换关键词
  \SetKwInOut{Output}{\textbf{Output}}
  \caption{cDNMF. }
  \Input{
         Text corpus $\mathbf{D}$,
          number of topics $T$,
         hyperparameters $\delta \geq  0$, $M \geq 0$, $\lambda_1 \geq 0$, and $\lambda_2 \geq 0$. \\
        }
    \Output{$\mathbf{C}^{(t)}$, $f(\mathbf{D})$.}
    \BlankLine
    Initialize: topic-word distribution $\mathbf{C}^0$,
          document-topic distribution $\W^0$,
          weight matrix $\mathbf{T}^0$, $ t \leftarrow 0 $;\\
    Construct a document-topic distribution $f(\mathbf{D})$ by deep
unsupervised learning methods; \\
    \Repeat{\text{convergence}}
        {
          Calculate $\mathbf{W}^{(t)}$ by \eqref{eq:MU_SED_W};\\
          Calculate $\mathbf{C}^{(t)}$ by \eqref{eq:MU_SED_C}; \\
          Calculate $\mathbf{T}^{(t)}$ by \eqref{eq:MU_SED_T}; \\
          $ t \leftarrow t+1 $; \\
          }

\end{algorithm}

bDNMF and sDNMF intrinsically assumes that each document contains only one topic, which may not be true.
To overcome the weakness of bDNMF and sDNMF, we propose cDNMF which introduces $f(\D)$ as a regularization on $\W$ instead of masking $\W$ by $f(\D)$ directly.

Specifically, we implement the discriminator $g(\W,f(\mathbf{D}))$ as a real-valued regression response of the semantic topic labels $f(\mathbf{D})$:
\begin{equation}\label{eq:mDNMFyyy}
  \begin{array}{c}
  \min\limits_{\mathbf{T}\in\mathbb{R}^{K\times K}, \W\geq\mathbf{0}} \| f(\mathbf{D}) - \mathbf{T}\mathbf{W} \|_F^2
  \end{array}
\end{equation}
where $\mathbf{T}$ denotes a linear transform of $\W$. To further constrains the word-topic matrix $\mathbf{C}$ for highly meaningful topic words, we propose a word-word affinity regularization $\Omega(\mathbf{C})$:
\begin{equation}\label{eq:mDNMFxxx}
  \Omega(\mathbf{C}) = \| \mathbf{C} \mathbf{C}^{T} - \mathbf{D} \mathbf{D}^{T} \|_F^2
\end{equation}
which encodes the word-word semantics from the shared knowledge between the documents into $\mathbf{C}$. To our knowledge, this is the first time that such a regularization is introduced to the NMF topic modeling.

Substituting \eqref{eq:mDNMFyyy} and \eqref{eq:mDNMFxxx} into the DNMF framework derives the objective of cDNMF:
%\begin{equation}\label{eq:cdnmf}
%\begin{split}
  %D_F(\mathbf{D}||\mathbf{C}\W,\mathbf{T},f(\mathbf{D})) = & \| \mathbf{D} - \mathbf{C}\W\|_F^2 + %  \lambda_2 \Omega(\mathbf{C}) \\
%  & \lambda_1 \| f(\mathbf{D}) - \mathbf{T}\W \|_F^2
%\end{split}
%\end{equation}
% \setlength{\arraycolsep}{0.0em}
\begin{equation}\label{eq:mDNMF}
 \min\limits_{\mathbf{T} \geq \mathbf{0},\mathbf{C} \geq \mathbf{0},\mathbf{W} \geq \mathbf{0},f(\cdot)} D_F(\mathbf{D}||\mathbf{C}\W,\mathbf{T},f(\mathbf{D}))
\end{equation}
where
\begin{equation}
\begin{array}{ll}
  D_F(\mathbf{D}||\mathbf{C}\W,\mathbf{T},f(\mathbf{D}))  \\
  = \| \mathbf{D} - \mathbf{C}\W\|_F^2
 + \lambda_1 \| f(\mathbf{D}) - \mathbf{T}\W \|_F^2 + \lambda_2 \Omega(\mathbf{C}) \\
   \end{array}
\end{equation}
with $\lambda_1$ and $\lambda_2$ as two hyperparameters.

Like bDNMF, we solve cDNMF by first obtaining $f(\D)$ from a deep model, and taking $f(\D)$ as a constant of \eqref{eq:mDNMF}. Then, we optimize \eqref{eq:mDNMF} for $\mathbf{C}$, $\W$, and $\mathbf{T}$ by the alterative least squares optimization algorithm. When $f(\D)$ is given, problem \eqref{eq:mDNMF} satisfies the following first-order KKT optimality conditions:
\begin{eqnarray}
\left\{ \begin{array}{ll}
\mathbf{C} \geq \mathbf{0}, \mathbf{W} \geq \mathbf{0}, \mathbf{T} \geq \mathbf{0} \vspace{1ex} \\
\frac{\partial{D_F(\mathbf{D}||\mathbf{C}\W,\mathbf{T},f(\mathbf{D}))}}{\partial\mathbf{C}} \geq \mathbf{0} \vspace{1ex} \\
\mathbf{C}\odot\frac{\partial{ D_F(\mathbf{D}||\mathbf{C}\W,\mathbf{T},f(\mathbf{D}))}}{\partial\mathbf{C}} =\mathbf{0} \vspace{1ex} \\
\frac{\partial{D_F(\mathbf{D}||\mathbf{C}\W,\mathbf{T},f(\mathbf{D}))}}{\partial\mathbf{W}} \geq \mathbf{0} \vspace{1ex} \\
\mathbf{W}\odot\frac{\partial{D_F(\mathbf{D}||\mathbf{C}\W,\mathbf{T},f(\mathbf{D}))}}{\partial\mathbf{W}} =\mathbf{0} \vspace{1ex} \\
\frac{\partial{D_F(\mathbf{D}||\mathbf{C}\W,\mathbf{T},f(\mathbf{D}))}}{\partial\mathbf{T}} \geq \mathbf{0} \vspace{1ex} \\
\mathbf{T}\odot\frac{\partial{D_F(\mathbf{D}||\mathbf{C}\W,\mathbf{T},f(\mathbf{D}))}}{\partial\mathbf{T}} =\mathbf{0} \vspace{1ex} \\
\end{array}\right.\label{eq:KKT_SED}
\end{eqnarray}
which guarantees that the optimization of \eqref{eq:mDNMF} converges to a stationary point.
Let $\bm\Psi$, ${\mathcal{\mathbf{Q}}}$, and $\mathcal{\mathbf{P}}$ be the Lagrange multipliers of the constraints $\mathbf{C} \geq \mathbf{0}$, $\mathbf{W} \geq \mathbf{0}$ and $\mathbf{T} \geq \mathbf{0}$, respectively. The Lagrangian $\mathcal{J}$ of \eqref{eq:mDNMF} is
\begin{equation}\label{eq:Lagrange_SED}
\begin{split}
\mathcal{J} = & \mbox{Tr}(\mathbf{D}\mathbf{D}^T)-2\mbox{Tr}(\mathbf{D}\W^T\mathbf{C}^T) + \mbox{Tr}(\mathbf{C}\W\W^T\mathbf{C}^T) \\
& + \lambda_1\mbox{Tr}(f(\mathbf{D})f^T(\mathbf{D}))  - 2\lambda_1\mbox{Tr}(f(\mathbf{D})\W^T\mathbf{T}^T) \\
& + \lambda_1\mbox{Tr}(\mathbf{T}\mathbf{W}\mathbf{W}^T\mathbf{T}^T) + \lambda_2\mbox{Tr}(\mathbf{D}\mathbf{D}^T\mathbf{D}\mathbf{D}^T) \\
& - 2\lambda_2\mbox{Tr}(\mathbf{D}\mathbf{D}^T\mathbf{C}\mathbf{C}^T) + \lambda_2\mbox{Tr}(\mathbf{C}\mathbf{C}^T\mathbf{C}\mathbf{C}^T) \\
& +  \mbox{Tr}(\mathbf{C}\bm\Psi) + \mbox{Tr}(\mathbf{W}\mathcal{\mathbf{Q}}) + \mbox{Tr}(\mathbf{T}\mathcal{\mathbf{P}}) \\
\end{split}
\end{equation}
The partial derivatives of $\mathcal{J}$ with respect to $\mathbf{C}$, $\mathbf{W}$ and $\mathbf{T}$ are
\begin{equation}\label{eq:dif_Lagrange_C}
\begin{split}
\frac{\partial \mathcal{J}}{\partial \mathbf{C}} = & - 2\mathbf{D}\mathbf{W}^T + 2\mathbf{C}\W\W^T -4\lambda_2\mathbf{D}\mathbf{D}^T\mathbf{C} + \\ & 4\lambda_2\mathbf{C}\mathbf{C}^T\mathbf{C} + \bm\Psi
\end{split}
\end{equation}
\begin{equation}\label{eq:dif_Lagrange_W}
\begin{split}
\frac{\partial \mathcal{J}}{\partial \mathbf{W}} = & -2\mathbf{C}^T\mathbf{D} + 2\mathbf{C}^T\mathbf{C}\W - 2\lambda_1\mathbf{T}^Tf(\mathbf{D})+ \\ & 2\lambda_1\mathbf{T}^T\mathbf{T}\W + \mathcal{\mathbf{Q}}
\end{split}
\end{equation}
\begin{equation}\label{eq:dif_Lagrange_T}
\begin{split}
\frac{\partial \mathcal{J}}{\partial \mathbf{T}} = -2\lambda_1f(\mathbf{D})\mathbf{W}^T + 2\lambda_1\mathbf{T}\W\W^T + \mathcal{\mathbf{P}}
\end{split}
\end{equation}
Using the KKT conditions $\mathbf{C}\odot\bm\Psi = \mathbf{0}$, $\mathbf{W}\odot\mathcal{\mathbf{Q}} = \mathbf{0}$ and $\mathbf{T}\odot\mathcal{\mathbf{P}} = \mathbf{0}$ we get the following update rule for $\mathbf{C}$:
\begin{equation}\label{eq:MU_SED_W}
\begin{split}
[\mathbf{W}]^{(t+1)}_{ij} \leftarrow [\mathbf{W}]^{(t)}_{ij}\frac{[\mathbf{C}^T\mathbf{D}]_{ij} + \lambda_1[\mathbf{T}^Tf(\mathbf{D})]_{ij}}{[\mathbf{C}^T\mathbf{C}\W]_{ij} + \lambda_1[\mathbf{T}^T\mathbf{T}\mathbf{W}]_{ij}}
\end{split}
\end{equation}
\begin{equation}\label{eq:MU_SED_C}
\begin{split}
[\mathbf{C}]^{(t+1)}_{ij} \leftarrow [\mathbf{C}]^{(t)}_{ij}\frac{[\mathbf{D}\mathbf{W}^T]_{ij} + 2\lambda_2[\mathbf{D}\mathbf{D}^T\mathbf{C}]_{ij}}{[\mathbf{C}\W\W^T]_{ij} + 2\lambda_2[\mathbf{C}\mathbf{C}^T\mathbf{C}]_{ij}}
\end{split}
\end{equation}
\begin{equation}\label{eq:MU_SED_T}
\begin{split}
[\mathbf{T}]^{(t+1)}_{ij} \leftarrow [\mathbf{T}]^{(t)}_{ij}\frac{[f(\mathbf{D})\mathbf{W}^T]_{ij}}{[\mathbf{T}\W\W^T]_{ij}}
\end{split}
\end{equation}

cDNMF is summarized in Algorithm \ref{alg:DNMF}. Its merit over bDNMF and sDNMF is that cDNMF avoids the assumption that each document contains only one topic. However, it has two tunable hyperparameters. As we know, there is no way to tune the hyperparameters in unsupervised topic modeling. To remedy this weakness, we take the document clustering result $f(\D)$ as the pseudo labels for tuning the hyperparameters.

\subsection{Deep unsupervised document clustering}\label{subsec:mbn}

In Section \ref{Background}, we have summarized the recent progress of unsupervised deep learning methods for document clustering. To our knowledge, the advantage of the deep learning based document clustering over conventional document clustering methods is not apparent in general. In this section, we propose a novel unsupervised deep learning based document clustering method, named MBN, to address this issue.
 %Here we emphasize that MBN is only a special implementation of the proposed DNMF. As will be shown in the experiments, other deep learning methods based on word-embeddings \cite{pmlr-v80-zhao18a} can boost the performance of DNMF as well.

\subsubsection{Algorithm description of MBN}

 MBN consists of $L$ gradually narrowed hidden layers from bottom-up. Each hidden layer consists of $M$ $k$-centroids clusterings ($M\gg 1$), where parameter $k$ at the $l$-th layer is denoted by $k_l$, $l=1,\ldots,L$. Each $k_l$-centroids clustering has $k_l$ output units, each of which indicates one cluster. The output layer is linear-kernel-based spectral clustering \cite{ng2002spectral}. We take the output of the spectral clustering as $f(\D)$.

 MBN is trained simply by stacking. To train the $l$-th layer, we simply train each $k_l$-centroids clustering as follows:
 \begin{itemize}
  \item \textbf{Random sampling of input}. The first step randomly selects $k_l$ documents from ${\mathbf{X}^{(l-1)} = [\mathbf{x}_{1}^{(l-1)},\ldots,\mathbf{x}_{N}^{(l-1)}]}$ as the $k_l$ centroids of the clustering. If $l=1$, then ${\mathbf{X}^{(l-1)}=\textbf{D}}$.
  \item \textbf{One-nearest-neighbor learning}. The second step assigns an input document $\mathbf{x}^{(l-1)}$ to one of the $k_l$ clusters by one-nearest-neighbor learning, and outputs a $k_l$-dimensional indicator vector ${\textbf{h}=[h_1,\dots,h_{k_l}]^T}$, which is a one-hot sparse vector indicating the nearest centroid to $\mathbf{x}^{(l-1)}$.
 \end{itemize}
 The output units of all $k_l$-centroids clusterings are concatenated as the input of their upper layer, i.e. $\mathbf{x}^{(l)} = [\mathbf{h}_1^T,\ldots,\mathbf{h}_M^T]^T$.  We use cosine similarity to evaluate the similarity between the input and the centroids in all layers.

  As described in \cite{zhang2018multilayer}, each layer of MBN is a histogram-based nonparametric density estimator, which does not make model assumptions on data; the hierarchical structure of MBN captures the nonlinearity of documents by building a vast number of hierarchical trees on the TF-IDF feature space implicitly.

 %\subsubsection{Fundamentals of MBN}

% MBN is an unsupervised deep ensemble learning method. Like many ensemble learning methods such as bootstrap aggregating, the estimation error of MBN is decreased linearly along with the increase of the number of the base clusterings, i.e. $M$. $M$ is usually set to a number of a few hundreds.
%
%MBN has a fundamental assumption that, for any $k$-centroids clustering, the small area around a centroid of the clustering is locally linear. Hyperparameter $k_1$ is important in balancing the estimation accuracy and computational complexity. The larger $k_1$ is set to, the more likely MBN captures the small local variations of data. However, the computational complexity is also increased linearly with respect to the increase of $k_1$. Moreover, when $k_1$ is particularly large, the model has a risk of overfitting to random noise. In practice, setting $k_1 = \lfloor N/2 \rfloor$ is generally good.

\subsubsection{Network structure of MBN}

The network structure of MBN is important to its effectiveness. First of all, we should set the hyperparameter $M$ to a large number, which guarantees the high estimation accuracy of MBN at each layer. Then, to maintain the tree structure and discriminability of MBN, we should set $\{k_l\}_{l=1}^L$ carefully by the following criteria:
% \begin{align}
% k_1 = \left\lfloor N/2 \right\rfloor\\
%\quad k_l = \left\lfloor \delta k_{l-1} \right\rfloor\\
% \begin{array}{l}
%k_L\mbox{ is set to guarantee that at least a document per}\\
% \mbox{class is chosen by a random sample in probability}
% \end{array}\label{eq:xx}
% \end{align}
  \begin{equation}
 k_1 = \left\lfloor N/2 \right\rfloor, \quad k_l = \left\lfloor \delta k_{l-1} \right\rfloor
 \end{equation}
  \begin{equation}
 \begin{array}{l}
k_L\mbox{ is set to guarantee that at least a document per}\\
 \mbox{class is chosen by a random sample in probability}
 \end{array}\label{eq:xx}
 \end{equation}
 where $\delta\in[0,1)$ is a user defined hyperparameter with $0.5$ as the default.

As analyzed in \cite{zhang2018multilayer}, the hyperparameter $\delta$ controls how aggressively the nonlinearity of data is reduced. If the data is highly nonlinear, then we set $\delta$ to a large number, which results in a very deep architecture; otherwise, we set $\delta$ to a small number. MBN is relatively sensitive to the selection of $\delta$. As will be shown in the experiment, setting $\delta = 0.5$ is safe, though tuning $\delta$ may lead to better performance.

The criterion \eqref{eq:xx} guarantees that each $k_L$-centroids clustering is a valid one in probability. Specifically, for any $k_L$-centroids clustering, if its centroids do not contain any document of a topic, then its output representation has no discriminability to the topic. To understand this point, we consider an extreme case: if $k_L = 1$, then the top hidden layer of MBN outputs the same representation for all documents. In practice, we implement \eqref{eq:xx} by:
  \begin{eqnarray}
 k_L \approx\left\{ \begin{array}{ll}
30 \lceil \frac{N}{N_{z}}\rceil, &\mbox{if $\mathbf{D}$ is strongly class imbalanced}\\
 1.5K,& \mbox{otherwise}\\
\end{array}\right.\label{eq:xx1}
 \end{eqnarray}
 %where $N_Z$ and $N_z$ are the document numbers of the largest and smallest topics respectively. If $N_Z$ and $N_z$ are unknown, we simply set $ k_L$ to a relatively large number, which needs some experience.
where $N_z$ is the size of the smallest topic. If $N_z$ is unknown, we simply set $ k_L$ to a number that is significantly larger than the number of topics, e.g. 300 or so.

\subsection{Discussions}

%Although the objective functions of bDNMF and cDNMF are different, bDNMF can be viewed as a special case of cDNMF.
%In bDNMF algorithm, we construct the topic-document matrix $\W$ by the output of deep model.
%The cDNMF attaches regularization that learn the latent distribution from the deep model as a regression response to topic-document matrix $\W$.
%However, in cDNMF, $\W$ is linearly constrained by the output of deep model.
%When the linear transform matrix $\mathbf{T}$ is an identity matrix, bDNMF and cDNMF are equivalent.

The DNMF variants are new in the NMF study even without the deep model.
First, the structured NMF component of sDNMF is different from existing structured NMF models. For example, nonsmooth NMF \cite{pascual2006nonsmooth} incorporates a smooth factor to make the basis matrix and coefficient matrix (i.e. the topic-document matrix and word-topic matrix respectively in topic modeling) sparse, and reconciles the contradiction between approximation and sparseness. Some other structured NMF methods \cite{gao2017local,li2017robust} adopt a global centroid for each basis vector to capture the manifold structure.
However, sDNMF takes the sparse representation of documents as a mask of the basis matrix. Second, although it is common to add regularization terms into the objective function of NMF, we did not observe the term \eqref{eq:mDNMFxxx} in the study of NMF. Although some similar form to \eqref{eq:mDNMFyyy} has been proposed in \cite{li2016robust} for hyperspectral unmixing, they learn the representation of data by a shallow model. Therefore, the objective function of cDNMF is fundamentally new to our knowledge.

Because sDNMF and cDNMF are non-convex optimization problems, we take the alternative iterative optimization algorithm to solve them. The convergence of the algorithm is guaranteed by the following theorem:
\begin{myTheo}\label{theorem1}
The objective values of sDNMF and cDNMF decreases monotonically and converges to a stationary point.
\end{myTheo}
\begin{proof}
See Appendix \ref{appendices_A} for the proof of Theorem \ref{theorem1} where we take cDNMF as an example.
The proof can be applied to sDNMF too whose objective value is non-increasing under the update rules \eqref{eq:MU_cDNMF_C} and \eqref{eq:MU_cDNMF_W}.
\end{proof}

\section{Experiments}
\label{Exp}
In this section, we compare the proposed DNMF with nine topic modeling methods on three benchmark text datasets.

%In all experiments, the fault parameters used in DNMF are set as follows:
% $a = 1$, $M=400$ $\delta=0.5$, and $\lambda=1/3$.\footnote{The default value of $\lambda$ is in the implementation of the ADMM algorithm.}

 \subsection{Data sets}

We conducted experiments on the 20-newsrgoups, top 30 largest topics of TDT2, and top 30 largest topics of Reuters-21578 document corpora.  20-Newsgroups consists of 18,846 documents with a vocabulary size of 26,214. This data set has 20 categories, each of which contains around 1,000 documents. The top 30 largest topics of TDT2 consists of 9,394 documents with a vocabulary size of 36,771 words. The top 30 largest topics of
 Reuters-21578 contains 8,293 documents in total with a vocabulary size of 18,933 words.

For TDT2 and Reuter-21578, we randomly selected 3 to 25 topics from the top 30 largest topics respectively for evaluation. For 20-newsgroups, we randomly selected 3 to 20 topics respectively for evaluation. For each comparison, we reported the average results over 50 Monte-Carlo runs.
The indices of the topics of the 50 independent runs on TDT2 are the same as those at \url{http://www.cad.zju.edu.cn/home/dengcai/Data/TextData.html}.
We extracted TF-IDF statistics from the bag-of-words model of the documents, and took cosine similarity to measure the similarity of two documents in the TF-IDF space.
%\begin{figure*}[!t]
%\centering
%\includegraphics[width=7.5in]{DNMF_20News.jpg}
%\caption{Visualizations of a subset(5000 documents) of 20Newsgroups produced by DNMF and 10 comparison methods.}
%\label{fig_corpus}
%\end{figure*}

%$\bullet$

%$\bullet$

%$\bullet$

%\subsection{Experimental settings}

%The numbers in bold denote the best performance.

% \footnote{The default value of $\lambda$ is in the implementation of the ADMM algorithm.}

 \subsection{Comparison algorithms}
The hyperparameters of DNMF in all experiments were set as follows: $M=400$, $\delta=0.5$, $\lambda_1 = 1$, and $\lambda_2 = 1$, unless otherwise stated. We compared DNMF with four probabilistic models \cite{papadimitriou2000latent,blei2003latent,cai2008modeling,cai2009probabilistic}, four NMF methods \cite{gillis2013fast,gillis2014successive,kumar2013fast,fu2018anchor}, and one deep learning based topic model \cite{henao2015deep} with their optimal hyperparameter settings. They are listed as follows:
\begin{itemize}
 \itemsep=0.0pt
  \item \textbf{Probabilistic latent semantic indexing (PLSI)}\cite{papadimitriou2000latent}.
  \item \textbf{Latent Dirichlet allocation (LDA)}\cite{blei2003latent}.
  \item \textbf{Laplacian probabilistic latent semantic indexing (LapPLSI)}\cite{cai2008modeling}.
  \item \textbf{Locally-consistent topic modeling (LTM)} \cite{cai2009probabilistic}.
  \item \textbf{Successive projection algorithm (SPA)}\cite{gillis2013fast}.
  \item \textbf{Successive nonnegative projection (SNPA)} \cite{gillis2014successive}.
  \item \textbf{A fast conical hull algorithm (XRAY)}\cite{kumar2013fast}.
  \item \textbf{Anchor-free correlated topic modeling (AchorFree)} \cite{fu2018anchor}.
  \item \textbf{Deep Poisson factor modeling (DPFA)} \cite{henao2015deep}:
    it is a deep learning based topic model built on the Dirichlet process.
    We set its DNN to a depth of two hidden layers, and set the number of the hidden units of the two hidden layers to $K$ and $\lceil K/2\rceil$ respectively. We used the output from the first hidden layer for clustering. The above setting results in the best performance.
  %  {\color{red}
  %\item \textbf{Word Embeddings Deep Topic Model (WEDTM)} \cite{pmlr-v80-zhao18a}:
  %    It is a deep topic model that leverages word embeddings to discover inter topic structures and intra topic structures.
  %    We set the number of topics in each layer to $K$. The other hyperparameters are default. We used 300-dimensional GloVe word embeddings pre-trained on Wikipedia\footnote{\url{https://nlp.stanford.edu/projects/glove/}}. We used the output $\theta$, which is the latent representation of documents, for clustering.}
\end{itemize}

 \subsection{Evaluation Metrics}

We evaluated the comparison results in terms of \textit{clustering accuracy} (ACC), \textit{coherence}, and \textit{similarity count} (SimCount). Clustering accuracy applies the hungarian algorithm\footnote{\url{http://www.cad.zju.edu.cn/home/dengcai/Data/code/hungarian.m}} to solve the permutation problem of predicted labels.\footnote{\url{http://www.cad.zju.edu.cn/home/dengcai/Data/code/bestMap.m}} Coherence evaluates the quality of a single topic by finding how many topic words belonging to the topic appear across the documents of the topic \cite{arora2013practical}:
    \begin{equation}\label{eq2}
    {\rm{Coh}}(\nu) = \sum_{v_1,v_2 \in {\nu}} \log{\frac{{\rm{freq}}(v_1,v_2)+\varepsilon}{{\rm{freq}}(v_2)}}
    \end{equation}
where $v_1$ and $v_2$ denote two words in the vocabulary, ${\rm{freq}}(v_1,v_2)$ denotes the number of the documents where $v_1$ and $v_2$ co-appear, ${\rm{freq}}(v_2)$ denotes the number of the documents containing $v_2$, and $\varepsilon = 0.01$ is used to prevent the input of the logarithm operator from zero. The higher the clustering accuracy or coherence score is, the better the topic model is.
%See \cite{fu2018anchor} for a good description of the three evaluation metrics.
Because the coherence measurement does not evaluate the redundancy of a topic, we used similarity count to measure the similarity between the topics. For each topic, similarity count is obtained simply by counting the number of the overlapped words in the leading $K$ words.
 The lower the similarity count score is, the better the topic model is.% {\color{blue}{\textbf{add reference of all evaluation metrics!}}}

\subsection{Main results}

\begin{figure}[t]
\vspace{-0.7cm} %调整图片与上文的垂直距离
\setlength{\abovecaptionskip}{0cm} %缩小caption和图像之间的距离
\setlength{\belowcaptionskip}{-0cm} %缩小caption和下方文字的距离
\centering
{\includegraphics[width=3.5in]{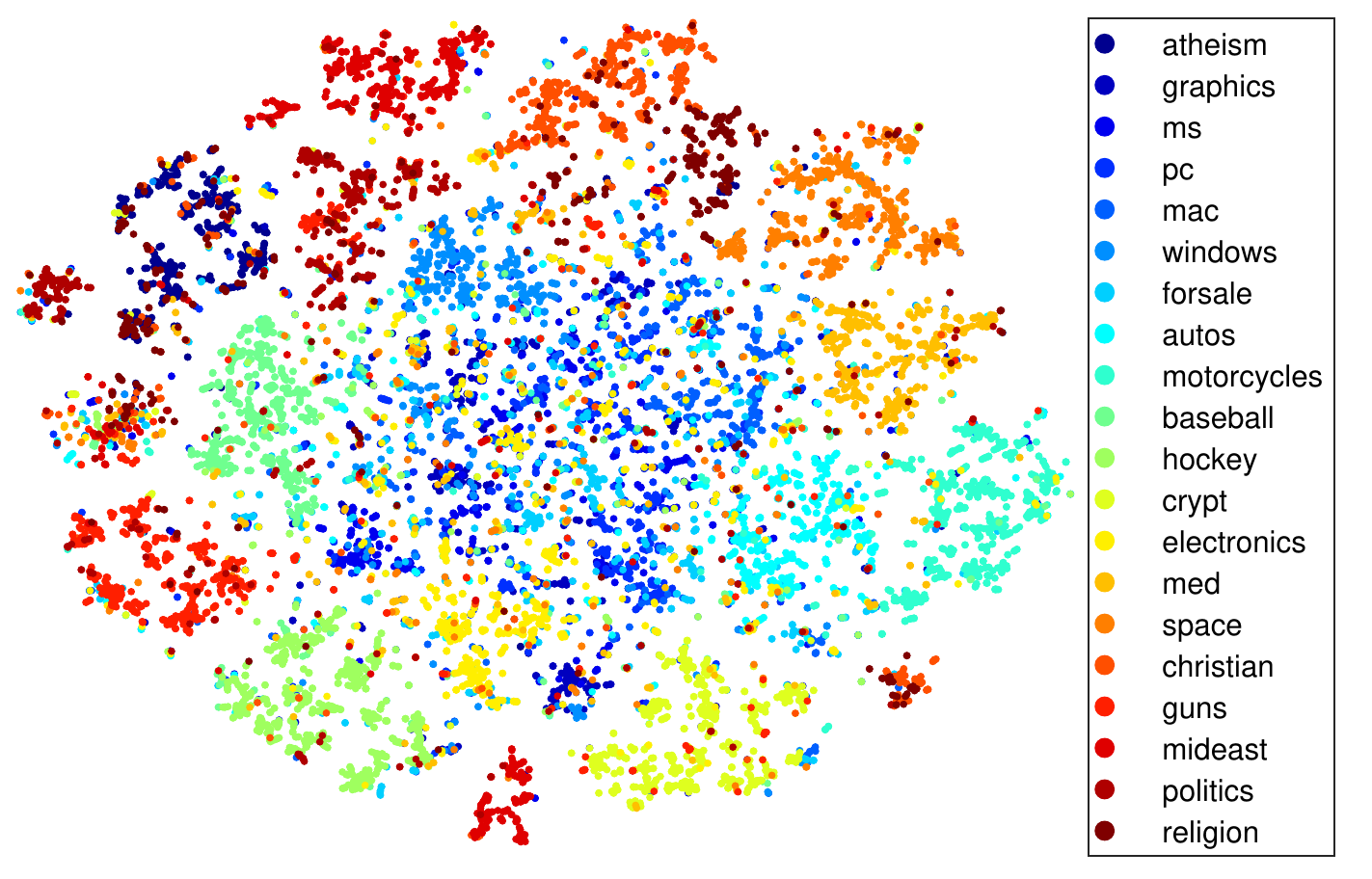}%
}
%\vspace{-0.7cm}
\caption{Visualizations of 20Newsgroups produced by DNMF.}
\label{fig_20News1}
\end{figure}

\begin{table}[tp]
\center
\caption{Topic words discovered by bDNMF and AnchorFree on a 5-topic subset of TDT2 corpus. The topic words in bold denotes overlapped words between topics.}\label{tab4}
\vspace{-0.3cm}
\tiny
\setlength{\tabcolsep}{0.1mm}{
\begin{tabular}%{|c|c|c|c|c|c|c|c|c|c|c|c|c|}
  {cccccccccc}
  %\hline
  %\cline{1-5}
  \multicolumn{5}{c}{AnchorFree}& \multicolumn{5}{c}{bDNMF}\\
  \toprule
  Topic 1 & Topic 2 & Topic 3 & Topic 4 & Topic 5 & Topic 1 & Topic 2 & Topic 3 & Topic 4 & Topic 5 \\
  \toprule

  netanyahu  & \textbf{asian}    & bowl      & tornadoes  & \textbf{economic}  &netanyahu   & asian      & bowl      & florida   & nigeria     \\
  israeli    & \textbf{asia}     & super     & florida    & \textbf{indonesia} &israeli     & percent    & super     & tornadoes & abacha   \\
  israel     & \textbf{economic} & broncos   & central    & \textbf{asian}     &israel      & indonesia  & broncos   & tornado   & military   \\
  palestinian& \textbf{financial}& denver    & storms     & \textbf{financial} &palestinian & asia       & denver    & storms    & police   \\
  peace      & \textbf{percent}  & packers   & ripped     & imf          & peace       & economy    & packers   & killed    & nigerian\\
  arafat     & \textbf{economy}  & bay       & victims    & \textbf{economy}  &albright    & financial  & green     & victims   & opposition   \\
  palestinians& market  & green     & tornado    & \textbf{crisis}     &arafat      & market     & game      & damage    & nigerias   \\
  albright   & stock    & football  & homes      & \textbf{asia}       &palestinians& stock      & bay       & homes     & anti  \\
  benjamin   & \textbf{crisis}   & game      & killed     & monetary   &talks       & economic   & football  & ripped    & elections   \\
  west       & markets  & san       & people     & \textbf{currency}   &west        & billion    & elway     & nino      & arrested   \\
  talks      & stocks   & elway     & damage     & \textbf{billion}    &benjamin    & crisis     & san       & el        & lagos  \\
  \textbf{bank}       & \textbf{currency} & diego     & twisters   & fund  &madeleine   & imf        & team      & weather   & democracy       \\
  prime      & prices   & xxxii     & nino       & \textbf{percent}    &london      & japan      & sports    & twisters  & sani  \\
  london     & dollar   & nfl       & el         & international        &ross        & spkr       & diego     & storm     & sysciviliantem\\
  minister   & investors& quarterback& deadly    & government           &withdrawal  & currency   & coach     & rain      & protest\\
  yasser     & index    & sports    & storm      & \textbf{bank}        &process     & markets    & play      & stories   & protests  \\
  ross       & \textbf{billion}  & play      & counties   & korea       &prime       & dollar     & win       & deadly    & presidential \\
  withdrawal & \textbf{bank}     & yards     & weather    & south       &yasser      & south      & teams     & struck    & abachas \\
  madeleine  & growth   & favre     & funerals   & indonesian       &secretary   & government & season    & residents & violent\\
  13         & \textbf{indonesia}& pittsburgh& toll       & suharto    &13          & prices     & fans      & california& nigerians  \\
  \toprule

\end{tabular}
}
\end{table}

\begin{table}[t]
\centering
\setlength{\abovecaptionskip}{0.cm}
\caption{Performance of the comparison algorithms on 20-newsgroups. }\label{tab1}
\tiny
\setlength{\tabcolsep}{0.1mm}{
\begin{tabular}{ccccccccccccccc}
  \toprule
   %after \\: \hline or \cline{col1-col2} \cline{col3-col4} ...
    & \#topics&PLSI&LDA&LapPLSI&LTM&SPA&SNPA&XRAY&AnchorFree&DFPA&bDNMF&sDNMF&cDNMF \\
   \toprule
   \multirow{9}{*}{ACC}
          & 3     & 0.4243  & 0.8134  & 0.7596  & 0.8955  & 0.4279  & 0.4275  & 0.4086  & 0.8763  & 0.8402  & 0.9101  & 0.9101  & 0.9101  \\
          & 4     & 0.3671  & 0.7291  & 0.7094  & 0.8287  & 0.3311  & 0.3312  & 0.3279  & 0.8360  & 0.8050  & 0.8916  & 0.8916  & 0.8916  \\
          & 5     & 0.3403  & 0.7013  & 0.6442  & 0.8389  & 0.2897  & 0.2900  & 0.2796  & 0.7618  & 0.7882  & 0.8689  & 0.8689  & 0.8689  \\
          & 6     & 0.3216  & 0.6622  & 0.6078  & 0.8229  & 0.2585  & 0.2585  & 0.2523  & 0.7095  & 0.7582  & 0.8549  & 0.8549  & 0.8549  \\
          & 7     & 0.3200  & 0.6462  & 0.6017  & 0.7881  & 0.2436  & 0.2437  & 0.2399  & 0.7132  & 0.7388  & 0.8266  & 0.8266  & 0.8266  \\
          & 8     & 0.3075  & 0.6178  & 0.5758  & 0.7744  & 0.2202  & 0.2203  & 0.2128  & 0.6888  & 0.7114  & 0.8056  & 0.8056  & 0.8056  \\
          & 9     & 0.3113  & 0.6021  & 0.5433  & 0.7207  & 0.2126  & 0.2123  & 0.1962  & 0.6622  & 0.6897  & 0.7662  & 0.7662  & 0.7662  \\
          & 10    & 0.3111  & 0.5915  & 0.5279  & 0.7107  & 0.2066  & 0.2069  & 0.1957  & 0.6431  & 0.6664  & 0.7584  & 0.7584  & 0.7584  \\
          & 15    & 0.3212  & 0.5187  & 0.4799  & 0.6328  & 0.1757  & 0.1756  & 0.1591  & 0.5208  & 0.5754  & 0.6860  & 0.6860  & 0.6860  \\
          & 20    & 0.3603  & 0.4900  & 0.4354  & 0.5996  & 0.1469  & 0.1475  & 0.1071  & 0.4465  & 0.5233  & 0.6502  & 0.6502  & 0.6502  \\
  \toprule
   rank&          & 9.2   & 6.9   & 8     & 4.1   & 10.5  & 10.3  & 12    & 5.8   & 5.2   & 1     & 1     & 1 \\
  \toprule

  \multirow{9}{*}{Coherence}
          & 3     & -963.64  & -603.30  & -725.59  & -636.46  & -558.28  & -558.28  & -980.59  & -572.86  & -534.39  & -667.66  & -635.60  & -694.15  \\
          & 4     & -1008.22  & -634.13  & -732.05  & -677.52  & -613.12  & -613.12  & -1076.65  & -573.30  & -585.47  & -666.69  & -659.37  & -702.69  \\
          & 5     & -995.86  & -651.71  & -739.20  & -704.39  & -618.71  & -618.71  & -1085.22  & -565.88  & -562.91  & -676.85  & -650.67  & -716.23  \\
          & 6     & -1003.28  & -678.72  & -743.40  & -753.48  & -650.35  & -650.35  & -1115.20  & -538.75  & -588.59  & -710.18  & -679.94  & -765.19  \\
          & 7     & -994.24  & -686.47  & -754.92  & -741.03  & -709.04  & -709.04  & -1156.77  & -544.18  & -587.15  & -695.82  & -664.72  & -744.81  \\
          & 8     & -1015.14  & -702.15  & -779.04  & -778.82  & -697.90  & -697.90  & -1215.27  & -566.80  & -592.84  & -704.10  & -677.69  & -771.84  \\
          & 9     & -1020.01  & -716.64  & -773.21  & -790.30  & -725.18  & -725.18  & -1200.38  & -562.41  & -605.59  & -711.80  & -674.26  & -771.47  \\
          & 10    & -1008.06  & -729.28  & -789.99  & -799.53  & -766.64  & -766.64  & -1236.49  & -571.91  & -616.24  & -721.11  & -688.32  & -794.68  \\
          & 15    & -1001.83  & -762.76  & -843.94  & -854.70  & -816.32  & -816.32  & -1335.53  & -575.56  & -640.45  & -749.68  & -713.89  & -866.38  \\
          & 20    & -911.33  & -759.13  & -856.07  & -855.94  & -901.74  & -901.74  & -1141.56  & -596.09  & -676.11  & -789.63  & -716.90  & -937.74  \\
  \toprule
  rank&          &    10.9  & 5     & 9     & 8.4   & 4.7   & 4.7   & 12    & 1.4   & 1.8   & 5.8   & 4.1   & 9.2 \\
  \toprule

  \multirow{9}{*}{SimCount}
          & 3     & 332.26  & 9.42  & 453.02  & 0.40  & 5.46  & 5.46  & 4.38  & 10.42  & 21.6  & 3.52  & 2.08  & 2.78  \\
          & 4     & 333.08  & 14.56  & 488.92  & 1.04  & 12.16  & 12.16  & 8.16  & 21.32  & 38.36 & 7.12  & 4.24  & 4.80  \\
          & 5     & 319.66  & 22.34  & 430.96  & 1.60  & 13.24  & 13.24  & 13.32  & 32.22  & 64.22 & 10.08  & 4.98  & 5.98  \\
          & 6     & 295.40  & 33.12  & 432.10  & 2.32  & 19.74  & 19.74  & 17.68  & 53.12  & 95.66 & 16.50  & 7.74  & 8.48  \\
          & 7     & 317.60  & 37.22  & 388.34  & 4.28  & 22.28  & 22.28  & 22.92  & 76.14  & 124.12 & 23.44  & 11.68  & 14.60  \\
          & 8     & 309.00  & 42.50  & 319.68  & 5.18  & 24.08  & 24.08  & 31.92  & 112.14  & 160.72 & 30.52  & 15.24  & 15.32  \\
          & 9     & 324.14  & 50.48  & 386.22  & 6.44  & 28.12  & 28.12  & 40.52  & 142.02  & 193.12 & 41.10  & 21.42  & 20.48  \\
          & 10    & 323.34  & 66.38  & 278.66  & 9.18  & 23.58  & 23.58  & 46.80  & 195.76  & 224.76 & 51.02  & 25.44  & 21.82  \\
          & 15    & 352.90  & 116.20  & 186.22  & 19.56  & 21.46  & 21.46  & 106.06  & 598.82  & 365.64 & 111.12  & 53.48  & 36.20  \\
          & 20    & 396.00  & 196.00  & 139.00  & 26.00  & 15.00  & 15.00  & 211.00  & 1235.00  & 496.12 & 189.96  & 89.44  & 56.82  \\
  \toprule
  rank &          &    10.9  & 8     & 11    & 1.2   & 4.1   & 5.1   & 6.2   & 9.6   & 10.2  & 5.7   & 3     & 3 \\
  \toprule
\end{tabular}}
\end{table}

%We have compared DNMF with 9 representative topic modeling methods  \cite{papadimitriou2000latent,blei2003latent,cai2008modeling, cai2009probabilistic,gillis2013fast,gillis2014successive,kumar2013fast ,fu2018anchor,henao2015deep} covering hierarchical probabilistic models \cite{papadimitriou2000latent,blei2003latent,cai2008modeling
%,cai2009probabilistic}, NMF methods  \cite{gillis2013fast,gillis2014successive,kumar2013fast,fu2018anchor}, and deep topic models  \cite{henao2015deep}.
%Empirical results on the 20-newsgroups, topic detection and tracking database version 2 (TDT2), and Reuters-21578 corpora illustrate the effectiveness of DNMF in terms of four evaluation metrics. Moreover, DNMF is insensitive to the hyperparameter selection, which facilitates its practical use. Here we first illustrate a comparison example between DNMF and the recent advanced AnchorFree algorithm \cite{fu2018anchor}, on a 5-topic subset of the TDT2 corpus in Table \ref{tab4}.
%From the table, we see that DNMF produces more discriminative and less overlapped topic words than AnchorFree.

We listed the top 20 topic words of a 5-topic modeling problem as an example in Table \ref{tab4}. From the table, we see that the topic words of the second and fifth topics produced by AnchorFree have an overlap of over 50\%. Some informative topic words discovered by bDNMF, such as the words on anti-government activities or violence in the fifth topic, were not detected by AnchorFree.
The above phenomena are observed in the other experiments too.
We conjecture that the advanced experimental phenomena are caused by the fact that bDNMF not only avoids making additional assumptions but also benefits from the high clustering accuracy of the deep model. We show the latent representation of the documents in 20-newsgroups learned by MBN in Fig. \ref{fig_20News1}. From the figure, we see that the latent representation has strong disriminability which may lead to high performance of DNMF.

Table \ref{tab1} shows the comparison results on the 20-newsgroups corpus. From the table, we see that the DNMF variants achieve the highest clustering accuracy among the comparison methods. For example, the clustering accuracy of DNMF is more than $5\%$ absolutely higher than that of the runner-up method, i.e. LTM, when the number of the topics is 20 and between 4 and 15, and is at least $1\%$ higher than the latter in the other cases. Particularly, DNMF is significantly better than the NMF methods. The relative improvement of DNMF over NMF tends to be enlarged when the number of the topics increases, which demonstrates the effectiveness of the deep architecture of DNMF. In addition, the single-topic quality produced by sDNMF ranks the third in terms of coherence, which is inferior to AnchorFree and DFPA.
The similarity count scores produced by sDNMF and cDNMF rank behind LTM and are higher than the other comparison methods, which indicates that DNMF is able to generate less overlapped topic words than the comparison methods except the probabilistic model---LTM.

\begin{table}[t]
\centering
\caption{Performance of the comparison algorithms on TDT2.}\label{tab2}
\vspace{-0.3cm}
\tiny
\setlength{\tabcolsep}{0.1mm}{
\begin{tabular}{ccccccccccccccc}
  \toprule
   %after \\: \hline or \cline{col1-col2} \cline{col3-col4} ...
  & \#topics&PLSI&LDA&LapPLSI&LTM&SPA&SNPA&XRAY&AnchorFree&DFPA&bDNMF&sDNMF&cDNMF \\
   \toprule
   \multirow{11}{*}{ACC}
          & 3     & 0.5497  & 0.7932  & 0.9889  & 0.9872  & 0.7853  & 0.7854  & 0.7263  & 0.9738  & 0.8840  & 0.9954  & 0.9954  & 0.9954  \\
          & 4     & 0.5187  & 0.7402  & 0.9831  & 0.9496  & 0.7291  & 0.7306  & 0.6782  & 0.9469  & 0.8151  & 0.9864  & 0.9864  & 0.9864  \\
          & 5     & 0.4939  & 0.7013  & 0.9771  & 0.9443  & 0.6943  & 0.6986  & 0.6716  & 0.9186  & 0.8037  & 0.9808  & 0.9808  & 0.9808  \\
          & 6     & 0.4678  & 0.6762  & 0.9683  & 0.9171  & 0.6452  & 0.6392  & 0.6267  & 0.9093  & 0.7888  & 0.9863  & 0.9863  & 0.9863  \\
          & 7     & 0.4814  & 0.6570  & 0.9392  & 0.8649  & 0.6125  & 0.6105  & 0.6253  & 0.9024  & 0.7490  & 0.9566  & 0.9566  & 0.9566  \\
          & 8     & 0.4721  & 0.6230  & 0.9457  & 0.8406  & 0.5818  & 0.5792  & 0.5736  & 0.8748  & 0.7233  & 0.9460  & 0.9460  & 0.9460  \\
          & 9     & 0.4930  & 0.6481  & 0.9095  & 0.8092  & 0.5883  & 0.5832  & 0.5433  & 0.8690  & 0.7469  & 0.9575  & 0.9575  & 0.9575  \\
          & 10    & 0.4883  & 0.6413  & 0.9017  & 0.7705  & 0.5642  & 0.5612  & 0.5319  & 0.8481  & 0.7305  & 0.9100  & 0.9100  & 0.9100  \\
          & 15    & 0.5412  & 0.5941  & 0.8393  & 0.6861  & 0.4736  & 0.4694  & 0.4411  & 0.7963  & 0.6849  & 0.8613  & 0.8613  & 0.8613  \\
          & 20    & 0.6290  & 0.6093  & 0.7606  & 0.6458  & 0.4593  & 0.4610  & 0.4358  & 0.7741  & 0.6776  & 0.8074  & 0.8074  & 0.8074  \\
          & 25    & 0.6582  & 0.6095  & 0.7390  & 0.6325  & 0.4351  & 0.4367  & 0.4240  & 0.7392  & 0.6521  & 0.7664  & 0.7664  & 0.7664  \\
  \toprule
rank  &          & 10.82  & 8.18  & 4.18  & 5.91  & 9.82  & 9.91  & 11.09  & 5.18  & 6.91  & 1.00  & 1.00  & 1.00  \\
  \toprule

  \multirow{11}{*}{Coherence}
          & 3     & -593.96  & -427.39  & -538.16  & -678.43  & -613.89  & -613.89  & -470.37  & -419.13  & -952.05  & -336.00  & -335.75  & -341.63  \\
          & 4     & -573.30  & -510.27  & -562.86  & -660.56  & -592.47  & -592.47  & -447.67  & -430.83  & -888.18  & -358.83  & -350.48  & -378.10  \\
          & 5     & -545.48  & -509.78  & -544.17  & -634.29  & -610.96  & -610.96  & -459.79  & -406.99  & -803.90  & -377.58  & -370.10  & -381.37  \\
          & 6     & -536.32  & -546.04  & -554.52  & -626.23  & -642.78  & -642.78  & -466.89  & -428.79  & -831.70  & -367.09  & -364.30  & -378.32  \\
          & 7     & -518.56  & -543.56  & -560.76  & -597.02  & -646.05  & -646.05  & -483.75  & -397.79  & -731.31  & -396.75  & -382.47  & -401.49  \\
          & 8     & -519.33  & -565.30  & -555.45  & -594.85  & -657.72  & -657.72  & -477.15  & -445.76  & -704.81  & -424.66  & -399.95  & -438.13  \\
          & 9     & -518.04  & -570.69  & -566.21  & -594.83  & -655.35  & -655.35  & -469.70  & -418.12  & -755.24  & -415.77  & -394.24  & -439.97  \\
          & 10    & -518.91  & -574.42  & -573.86  & -597.61  & -668.08  & -668.08  & -508.05  & -422.32  & -715.69  & -436.64  & -414.76  & -446.00  \\
          & 15    & -507.35  & -617.88  & -624.20  & -579.34  & -660.27  & -660.27  & -493.83  & -433.01  & -676.80  & -519.91  & -457.45  & -523.13  \\
          & 20    & -557.22  & -642.49  & -660.17  & -616.12  & -679.49  & -679.49  & -497.80  & -458.33  & -627.00  & -549.08  & -478.88  & -562.01  \\
          & 25    & -598.00  & -666.08  & -694.07  & -635.02  & -686.57  & -686.57  & -517.31  & -469.48  & -588.06  & -572.57  & -501.95  & -589.15  \\
  \toprule
rank  &          &    6.36  & 7.36  & 8.09  & 9.09  & 9.82  & 9.82  & 4.55  & 2.82  & 11.00  & 2.73  & 1.27  & 4.09  \\
  \toprule

  \multirow{11}{*}{SimCount}
          & 3     & 216.60  & 2.78  & 419.04  & 24.44  & 16.10  & 16.10  & 22.44  & 4.08  & 50.08  & 0.42  & 0.38  & 0.72  \\
          & 4     & 216.52  & 5.26  & 308.62  & 25.32  & 24.00  & 24.00  & 34.94  & 2.22  & 54.12  & 0.94  & 0.98  & 1.26  \\
          & 5     & 209.04  & 8.02  & 282.56  & 24.74  & 29.36  & 29.36  & 57.68  & 4.94  & 112.22  & 1.26  & 1.08  & 1.48  \\
          & 6     & 195.50  & 11.90  & 225.34  & 23.12  & 44.14  & 44.14  & 68.54  & 6.62  & 113.76  & 2.08  & 1.92  & 2.68  \\
          & 7     & 176.74  & 16.06  & 204.44  & 25.40  & 53.52  & 53.52  & 95.46  & 4.48  & 191.38  & 3.16  & 2.98  & 3.00  \\
          & 8     & 160.20  & 21.12  & 198.92  & 23.94  & 58.76  & 58.76  & 132.42  & 8.84  & 192.30  & 5.28  & 4.78  & 6.58  \\
          & 9     & 161.46  & 25.46  & 163.18  & 24.34  & 74.48  & 74.48  & 159.44  & 9.92  & 285.06  & 6.98  & 6.32  & 8.06  \\
          & 10    & 146.84  & 30.48  & 139.46  & 23.34  & 74.78  & 74.78  & 182.96  & 13.46  & 287.76  & 8.00  & 7.36  & 7.96  \\
          & 15    & 91.82  & 65.08  & 80.14  & 23.26  & 189.44  & 189.44  & 481.58  & 40.78  & 690.20  & 25.66  & 23.40  & 22.90  \\
          & 20    & 70.60  & 104.82  & 49.32  & 20.76  & 271.50  & 271.50  & 712.50  & 79.70  & 1056.20  & 50.22  & 44.90  & 43.64  \\
          & 25    & 52.18  & 147.22  & 33.20  & 22.78  & 450.14  & 450.14  & 936.52  & 132.34  & 1741.28  & 72.52  & 66.32  & 57.28  \\
  \toprule
 rank &          &    9.18  & 5.73  & 9.55  & 5.00  & 7.36  & 8.36  & 9.64  & 4.73  & 11.09  & 2.91  & 1.82  & 2.64  \\
  \toprule
\end{tabular}}
\end{table}

Table \ref{tab2} shows the results on the TDT2 corpus.
From the table, we see that the DNMF variants obtain the best performance in terms of clustering accuracy and similarity count, particularly when the number of topics is below 10. bDNMF and sDNMF outperform the comparison algorithms in terms of all three evaluation metrics, which demonstrates the advantage of the DNMF framework further.
Although LapPLSI yields competitive clustering accuracy with DNMF, its performance in coherence and similarity count is significantly lower than DNMF.
Although AnchorFree reaches a higher coherence rank than cDNMF, its similarity count scores are much higher than cDNMF.

\begin{table}[t]
\centering
\caption{Performance of the comparison algorithms on Reuters-21578.}\label{tab3}
\vspace{-0.3cm}
\tiny
\setlength{\tabcolsep}{0.1mm}{
\begin{tabular}{ccccccccccccccc}
  \toprule
   %after \\: \hline or \cline{col1-col2} \cline{col3-col4} ...
  & \#topics&PLSI&LDA&LapPLSI&LTM&SPA&SNPA&XRAY&AnchorFree&DFPA&bDNMF&sDNMF&cDNMF \\
   \toprule
   \multirow{11}{*}{ACC}
          & 3     & 0.6012  & 0.6269  & 0.7797  & 0.7445  & 0.7853  & 0.7854  & 0.7263  & 0.7904  & 0.7155  & 0.8591  & 0.8591  & 0.8591  \\
          & 4     & 0.5253  & 0.5691  & 0.6966  & 0.6870  & 0.7291  & 0.7306  & 0.6782  & 0.7257  & 0.6587  & 0.7745  & 0.7745  & 0.7745  \\
          & 5     & 0.4671  & 0.5290  & 0.6614  & 0.6198  & 0.6943  & 0.6986  & 0.6716  & 0.6480  & 0.6171  & 0.7160  & 0.7160  & 0.7160  \\
          & 6     & 0.4648  & 0.5140  & 0.6165  & 0.5844  & 0.6452  & 0.6392  & 0.6267  & 0.6449  & 0.6169  & 0.6803  & 0.6803  & 0.6803  \\
          & 7     & 0.4182  & 0.4628  & 0.6122  & 0.5820  & 0.6125  & 0.6105  & 0.6253  & 0.6472  & 0.5524  & 0.6948  & 0.6948  & 0.6948  \\
          & 8     & 0.4049  & 0.4442  & 0.6067  & 0.5663  & 0.5818  & 0.5792  & 0.5736  & 0.6133  & 0.5504  & 0.6474  & 0.6474  & 0.6474  \\
          & 9     & 0.3708  & 0.4064  & 0.5914  & 0.5490  & 0.5883  & 0.5832  & 0.5433  & 0.5886  & 0.5089  & 0.6244  & 0.6244  & 0.6244  \\
          & 10    & 0.3765  & 0.4150  & 0.5628  & 0.5279  & 0.5642  & 0.5612  & 0.5319  & 0.5822  & 0.5386  & 0.6110  & 0.6110  & 0.6110  \\
          & 15    & 0.3278  & 0.3545  & 0.4417  & 0.4210  & 0.4736  & 0.4694  & 0.4411  & 0.5198  & 0.5082  & 0.5189  & 0.5189  & 0.5189  \\
          & 20    & 0.3371  & 0.3331  & 0.4083  & 0.3624  & 0.4593  & 0.4610  & 0.4358  & 0.5294  & 0.4585  & 0.4899  & 0.4899  & 0.4899  \\
          & 25    & 0.3601  & 0.3373  & 0.3615  & 0.3553  & 0.4351  & 0.4367  & 0.4240  & 0.4684  & 0.4218  & 0.4702  & 0.4702  & 0.4702  \\
  \toprule
  rank&          & 11.73  & 11.18  & 7.09  & 9.27  & 5.55  & 5.91  & 7.82  & 4.18  & 8.73  & 1.18  & 1.18  & 1.18  \\
  \toprule

  \multirow{11}{*}{Coherence}
          & 3     & -769.73  & -674.14  & -852.48  & -943.56  & -613.89  & -613.89  & -470.37  & -827.28  & -996.30  & -760.47  & -647.97  & -759.97  \\
          & 4     & -786.89  & -677.18  & -813.51  & -952.28  & -592.47  & -592.47  & -447.67  & -743.97  & -1017.05  & -719.27  & -609.09  & -726.13  \\
          & 5     & -785.65  & -686.31  & -838.40  & -942.68  & -610.96  & -610.96  & -459.79  & -771.63  & -1045.24  & -752.04  & -620.00  & -746.92  \\
          & 6     & -805.24  & -715.15  & -854.96  & -947.58  & -642.78  & -642.78  & -466.89  & -699.50  & -1046.70  & -764.52  & -639.33  & -766.60  \\
          & 7     & -806.03  & -705.90  & -804.15  & -940.69  & -646.05  & -646.05  & -483.75  & -684.54  & -982.35  & -784.11  & -654.06  & -793.31  \\
          & 8     & -789.16  & -762.92  & -860.11  & -967.17  & -657.72  & -657.72  & -477.15  & -722.67  & -901.23  & -825.28  & -674.71  & -826.15  \\
          & 9     & -793.27  & -776.83  & -841.44  & -975.13  & -655.35  & -655.35  & -469.70  & -710.96  & -858.33  & -832.10  & -699.29  & -851.84  \\
          & 10    & -790.22  & -776.46  & -831.18  & -945.31  & -668.08  & -668.08  & -508.05  & -703.61  & -911.27  & -808.08  & -672.49  & -828.59  \\
          & 15    & -837.89  & -847.72  & -848.49  & -959.15  & -660.27  & -660.27  & -493.83  & -685.33  & -950.77  & -807.32  & -669.67  & -859.50  \\
          & 20    & -831.64  & -903.37  & -845.18  & -955.92  & -679.49  & -679.49  & -497.80  & -678.43  & -911.14  & -846.06  & -709.83  & -916.70  \\
          & 25    & -827.83  & -902.68  & -831.65  & -932.96  & -686.57  & -686.57  & -517.31  & -667.75  & -905.43  & -851.30  & -708.87  & -969.11  \\
  \toprule
  rank&           &    7.73  & 6.45  & 9.18  & 11.45  & 2.27  & 3.27  & 1.00  & 5.36  & 11.27  & 7.27  & 4.00  & 8.73  \\
  \toprule

  \multirow{11}{*}{SimCount}
          & 3     & 3.20  & 230.84  & 765.22  & 45.12  & 16.10  & 16.10  & 22.44  & 7.26  & 49.70  & 3.60  & 3.10  & 2.96  \\
          & 4     & 6.46  & 218.28  & 759.62  & 39.60  & 24.00  & 24.00  & 34.94  & 12.00  & 51.22  & 7.84  & 5.98  & 6.56  \\
          & 5     & 9.32  & 223.40  & 694.86  & 38.76  & 29.36  & 29.36  & 57.68  & 16.90  & 104.92  & 11.16  & 9.28  & 9.60  \\
          & 6     & 12.48  & 228.04  & 661.24  & 40.58  & 44.14  & 44.14  & 68.54  & 19.62  & 109.74  & 16.36  & 13.74  & 13.62  \\
          & 7     & 21.22  & 221.34  & 721.32  & 41.66  & 53.52  & 53.52  & 95.46  & 33.40  & 190.46  & 22.00  & 17.68  & 18.08  \\
          & 8     & 24.60  & 277.82  & 653.54  & 46.96  & 58.76  & 58.76  & 132.42  & 61.60  & 189.82  & 34.20  & 28.68  & 27.92  \\
          & 9     & 33.56  & 332.46  & 628.38  & 55.42  & 74.48  & 74.48  & 159.44  & 69.76  & 289.18  & 43.80  & 35.90  & 33.80  \\
          & 10    & 39.68  & 276.18  & 607.02  & 51.18  & 74.78  & 74.78  & 182.96  & 86.00  & 287.54  & 48.60  & 41.02  & 38.22  \\
          & 15    & 76.02  & 209.54  & 658.12  & 46.20  & 189.44  & 189.44  & 481.58  & 126.70  & 658.42  & 137.28  & 121.78  & 90.40  \\
          & 20    & 130.54  & 222.64  & 637.50  & 49.44  & 271.50  & 271.50  & 712.50  & 226.02  & 1000.46  & 227.28  & 198.32  & 137.78  \\
          & 25    & 194.98  & 202.88  & 615.52  & 48.94  & 450.14  & 450.14  & 936.52  & 339.68  & 1607.34  & 296.32  & 251.80  & 148.94  \\
  \toprule
  rank&          &    2.00  & 9.55  & 11.55  & 5.09  & 6.73  & 6.73  & 9.27  & 5.91  & 10.64  & 4.64  & 2.73  & 2.18  \\
  \toprule
\end{tabular}}
\end{table}

Table \ref{tab3} shows the performance of the comparison methods on the Reuters-21578 corpus. From the table, we see that the DNMF variants reach the highest clustering accuracy. Although it seems that they do not reach the top performance in terms of coherence and similarity count soly, they balance the coherence and similarity count which evaluate two contradict aspects of a topic model. For example, although the coherence of sDNMF ranks behind XRAY and SPA, its similarity count is much higher than the latter. Although the similarity count of sDNMF ranks behind PLSI, its coherence is higher than PLSI as well. If we average the coherence and similarity count ranking lists, it is clear that sDNMF performs the best.

 \begin{table}[t]
\centering
\tiny
\caption{Average ranks of the comparison methods on all three data sets. The ``Coh.+SimCount'' ranking list is the average of the lists in coherence and similarity count. The ``overall'' ranking list is the average of the lists in the three evaluation metrics.}\label{tab:avg_list}
\vspace{-0.3cm}
\setlength{\tabcolsep}{0.4mm}{
\begin{tabular}{lccccccccccccc}
  \toprule
   %after \\: \hline or \cline{col1-col2} \cline{col3-col4} ...
  &PLSI&LDA&LapPLSI&LTM&SPA&SNPA&XRAY&AnchorFree&DFPA&bDNMF&sDNMF&cDNMF \\
   \toprule
  ACC  &    10.58  & 8.75  & 6.42  & 6.43  & 8.62  & 8.71  & 10.30  & 5.05  & 6.95  & 1.06  & 1.06  & 1.06  \\
    Coherence&    8.33  & 6.27  & 8.76  & 9.65  & 5.60  & 5.60  & 5.85  & 3.19  & 8.02  & 5.27  & 3.12  & 7.34  \\
  SimCount&     7.36  & 7.76  & 10.70  & 3.76  & 6.06  & 6.06  & 8.37  & 6.75  & 10.64  & 4.42  & 2.52  & 2.61  \\
  \hline
  Coh.+SimCount&    7.85&	7.02&	9.73&	6.71&	5.83&	5.83&	7.11&	4.97&	9.33&	4.85&	2.82&	4.98\\
  \textbf{Overall}&    8.76  & 7.59  & 8.63  & 6.61  & 6.76  & 6.79  & 8.17  & 5.00  & 8.54  & 3.58  & 2.23  & 3.67  \\
  \toprule
\end{tabular}}
\end{table}

We summarize the ranking lists of the comparison methods on the three corpora in Table \ref{tab:avg_list}. From the overall ranking list in the table, we see that (i) the DNMF variants perform the best generally, followed by AnchorFree and LTM, and (ii) sDNMF performs the best among the three DNMF variants. If we take a look at the average ranking list over coherence and similarity count, we find that sDNMF reach the top performance, while bDNMF and sDNMF behave similarly with AnchorFree---a recent advanced NMF method that avoids the anchor-word assumption.

\begin{figure}[t]
\vspace{-0.7cm} %调整图片与上文的垂直距离
\setlength{\abovecaptionskip}{0cm} %缩小caption和图像之间的距离
\setlength{\belowcaptionskip}{-0cm} %缩小caption和下方文字的距离
\centering
\subfloat[]{\includegraphics[width=1.75in]{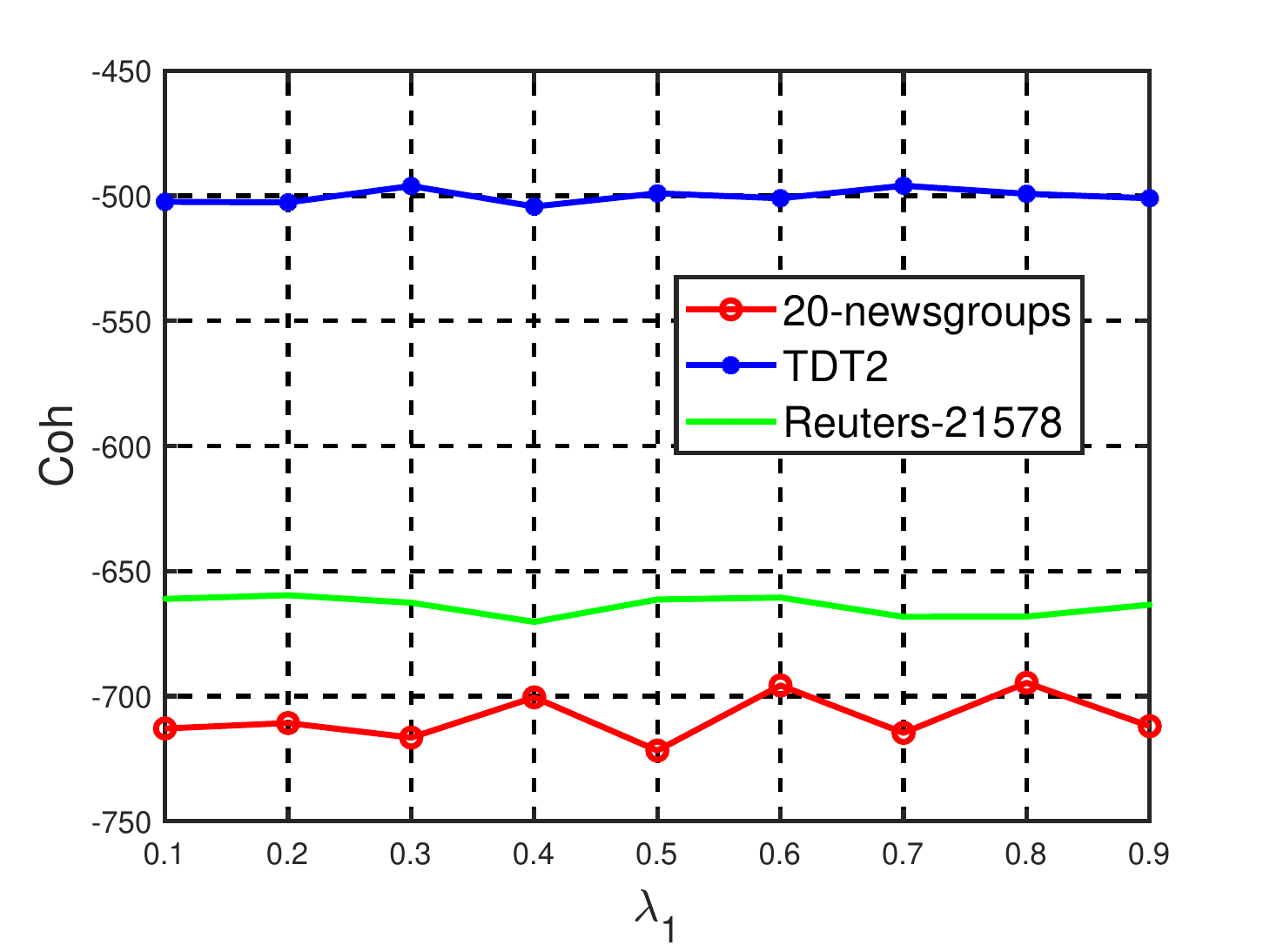}%
\label{lambda_1_Coh}}
\hfil
\subfloat[]{\includegraphics[width=1.75in]{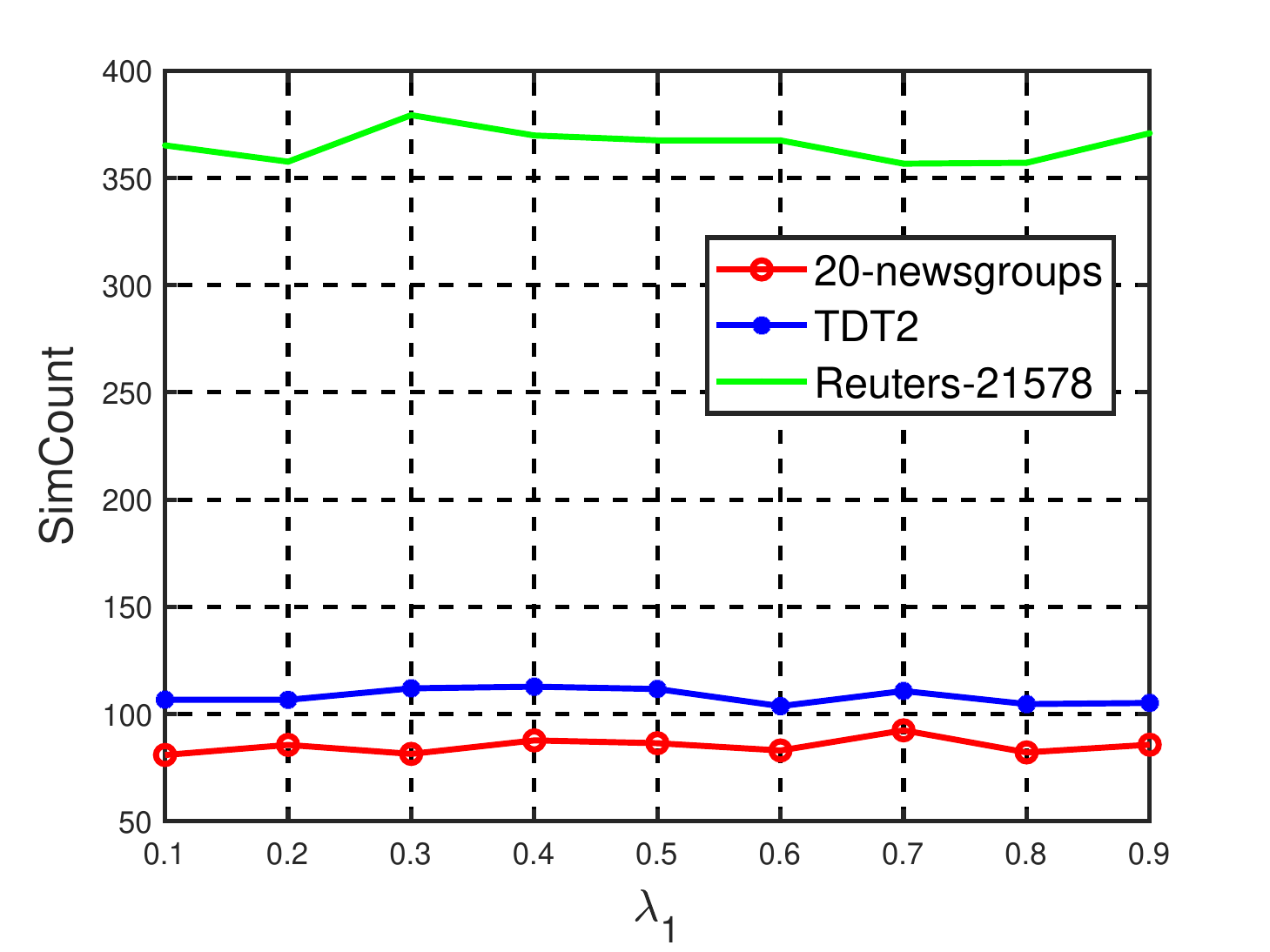}%
\label{lambda_1_SC}}
\caption{Performance of cDNMF with respect to hyperparameter $\lambda_1$ in terms of coherence and similarity count.}
\label{lambda_1}
\end{figure}

\begin{figure}[!t]
\vspace{-0.7cm} %调整图片与上文的垂直距离
\setlength{\abovecaptionskip}{0cm} %缩小caption和图像之间的距离
\setlength{\belowcaptionskip}{-0cm} %缩小caption和下方文字的距离
\centering
\subfloat[]{\includegraphics[width=1.75in]{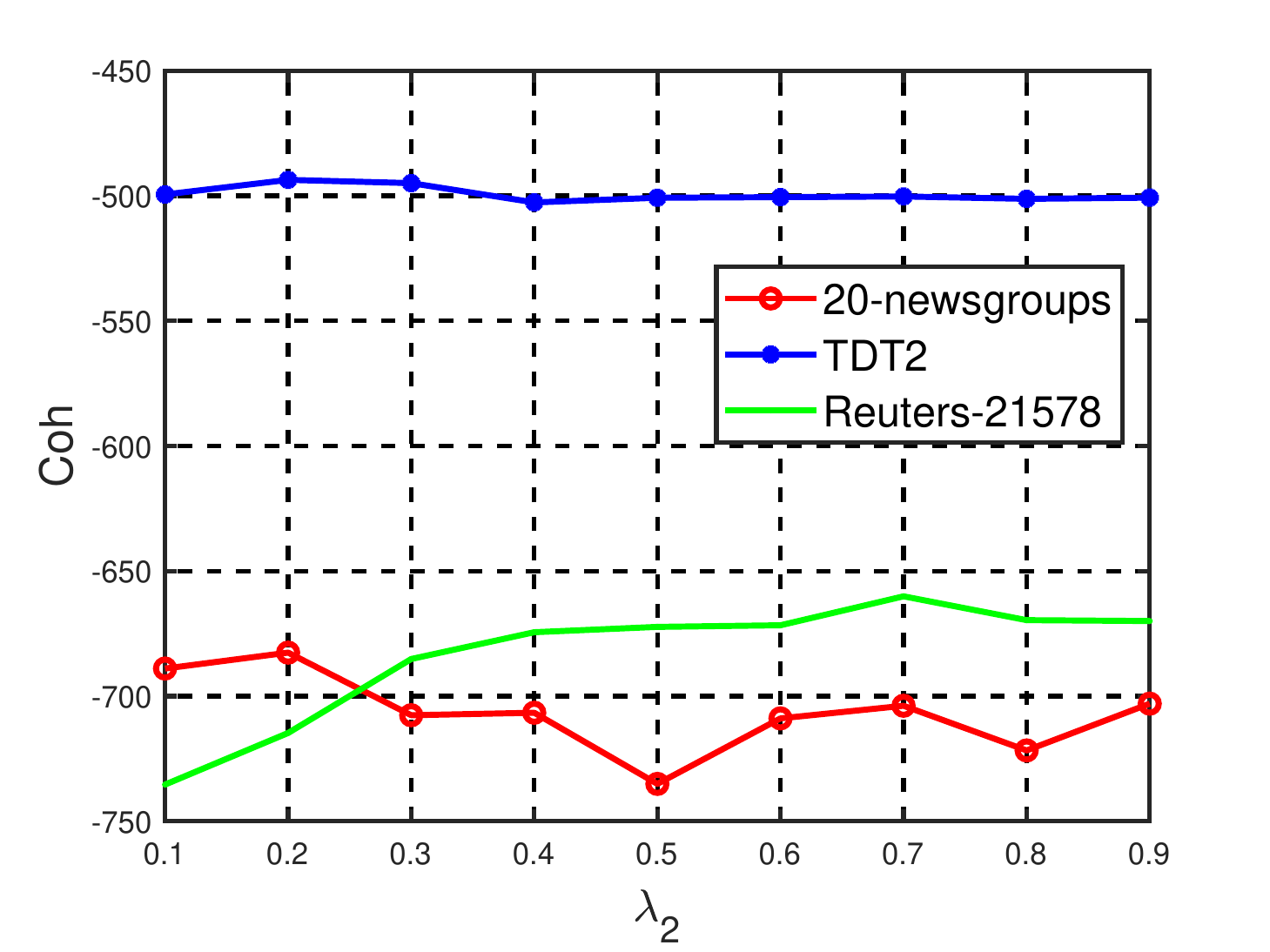}%
\label{lambda_2_Coh}}
\hfil
\subfloat[]{\includegraphics[width=1.75in]{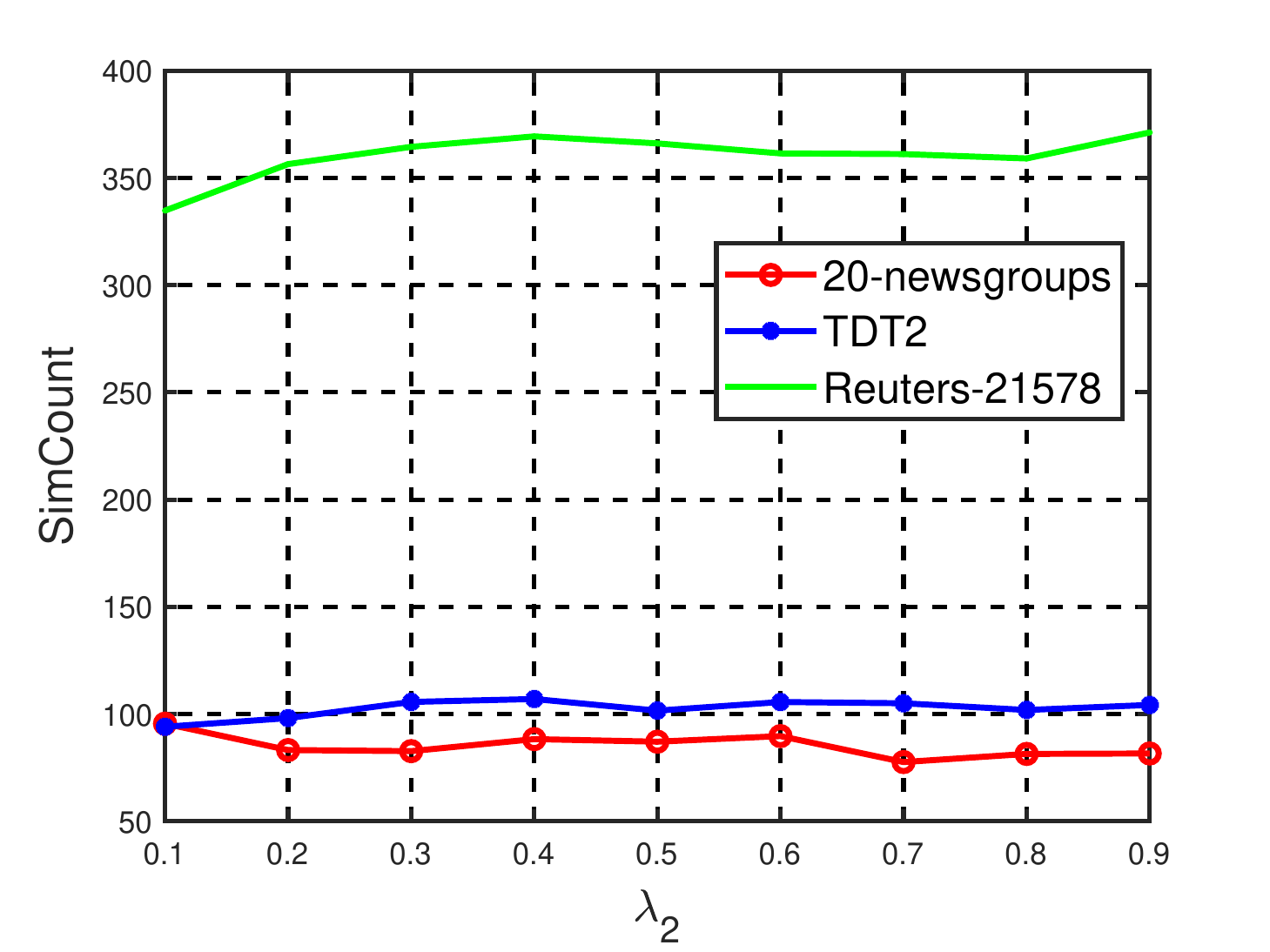}%
\label{lambda_2_SC}}
\caption{Performance of cDNMF with respect to hyperparameter $\lambda_2$.}
\label{lambda_2}
\end{figure}

\begin{figure}[!t]
\vspace{-0.7cm} %调整图片与上文的垂直距离
\setlength{\abovecaptionskip}{0cm} %缩小caption和图像之间的距离
\setlength{\belowcaptionskip}{-0cm} %缩小caption和下方文字的距离
\centering
\subfloat[]{\includegraphics[width=1.75in]{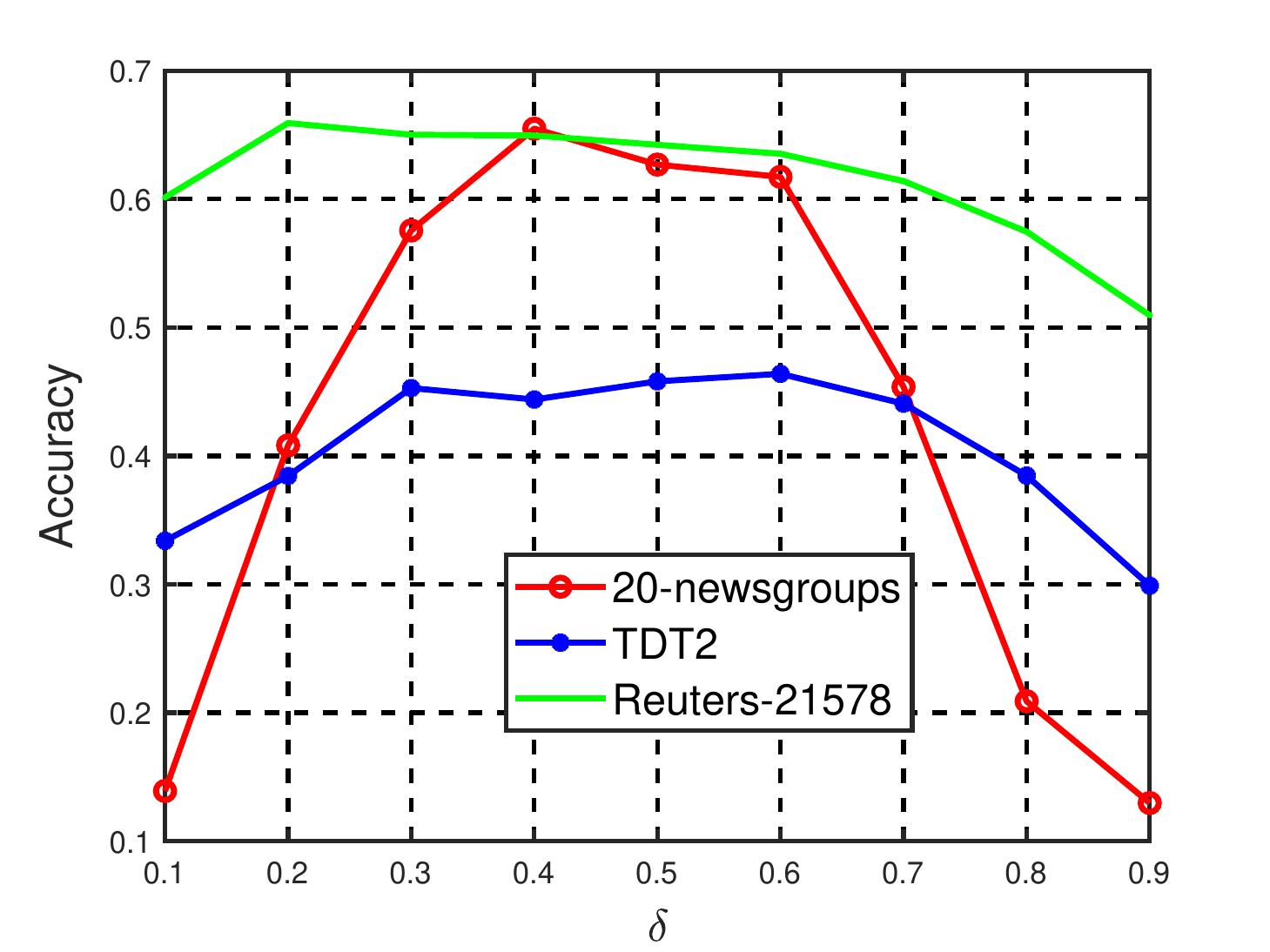}%
\label{ACC_delta}}
\hfil
\subfloat[]{\includegraphics[width=1.75in]{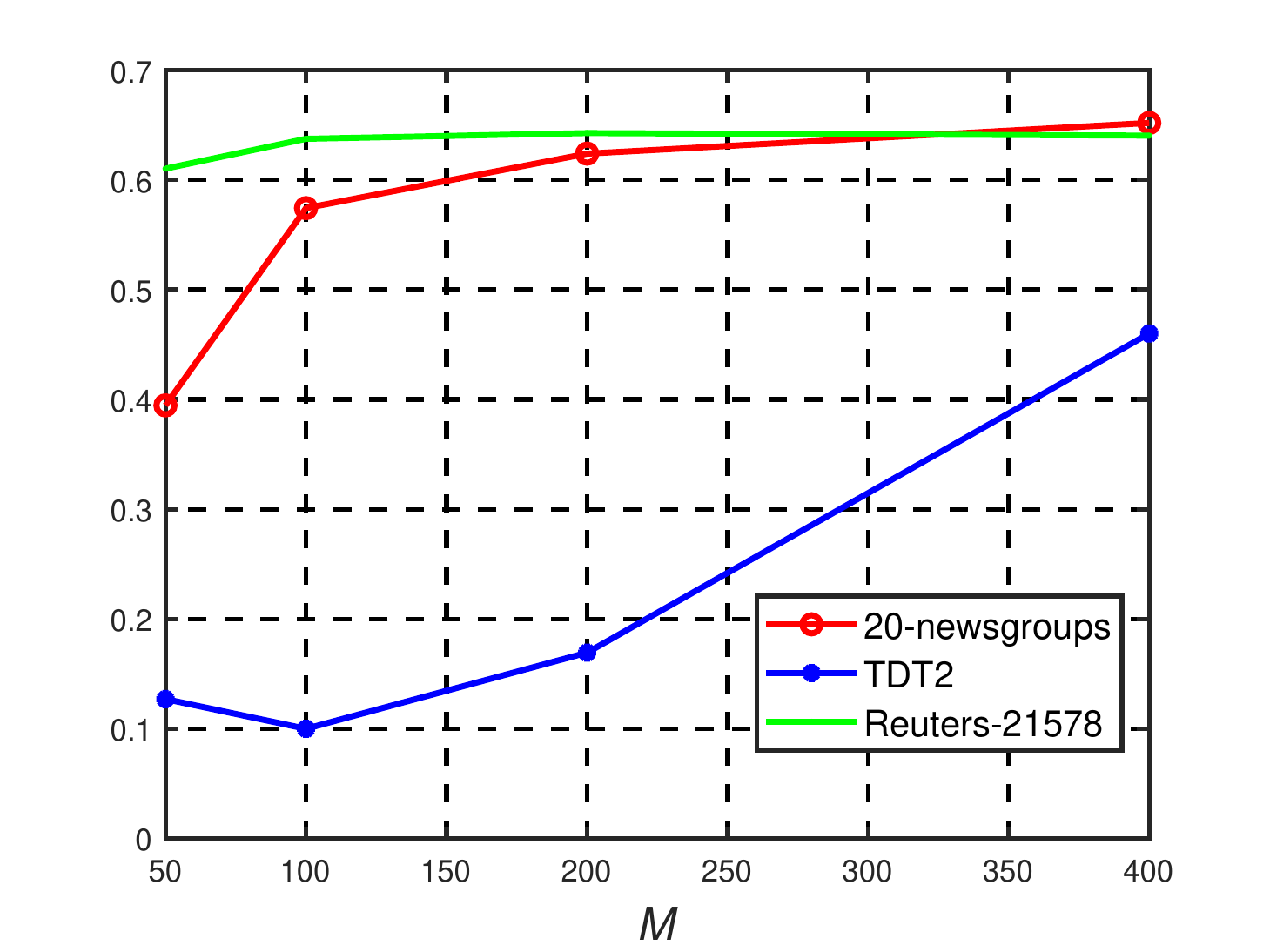}%
\label{ACC_V}}
\caption{Clustering accuracy of DNMF with respect to hyperparameters $\delta$ and $M$.}
\label{ACC}
\end{figure}

\begin{figure}[t]
\vspace{-0.7cm} %调整图片与上文的垂直距离
\setlength{\abovecaptionskip}{0cm} %缩小caption和图像之间的距离
\setlength{\belowcaptionskip}{-0cm} %缩小caption和下方文字的距离
\centering
\subfloat[]{\includegraphics[width=1.75in]{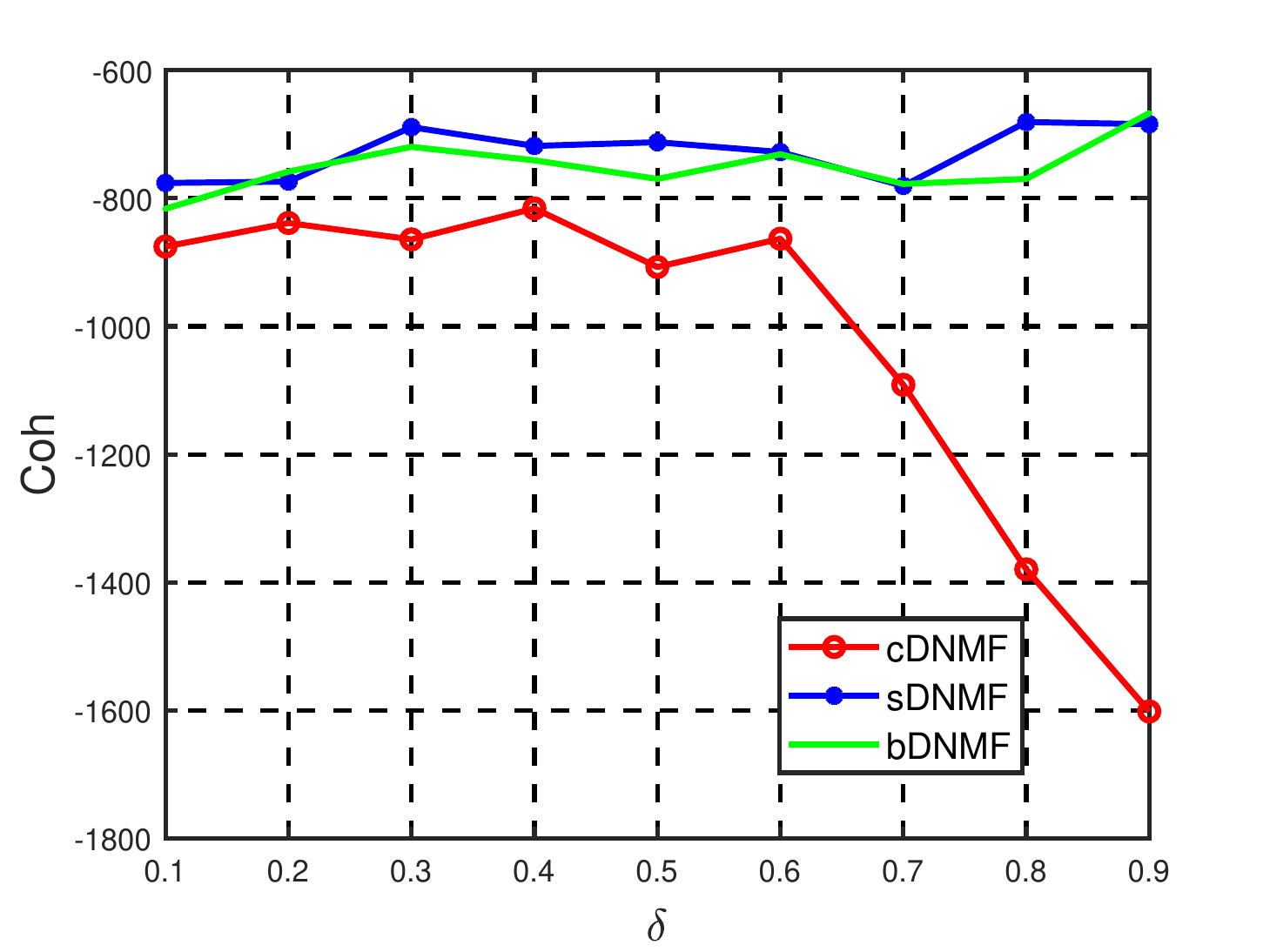}%
\label{20News_delta_Coh}}
\hfil
\subfloat[]{\includegraphics[width=1.75in]{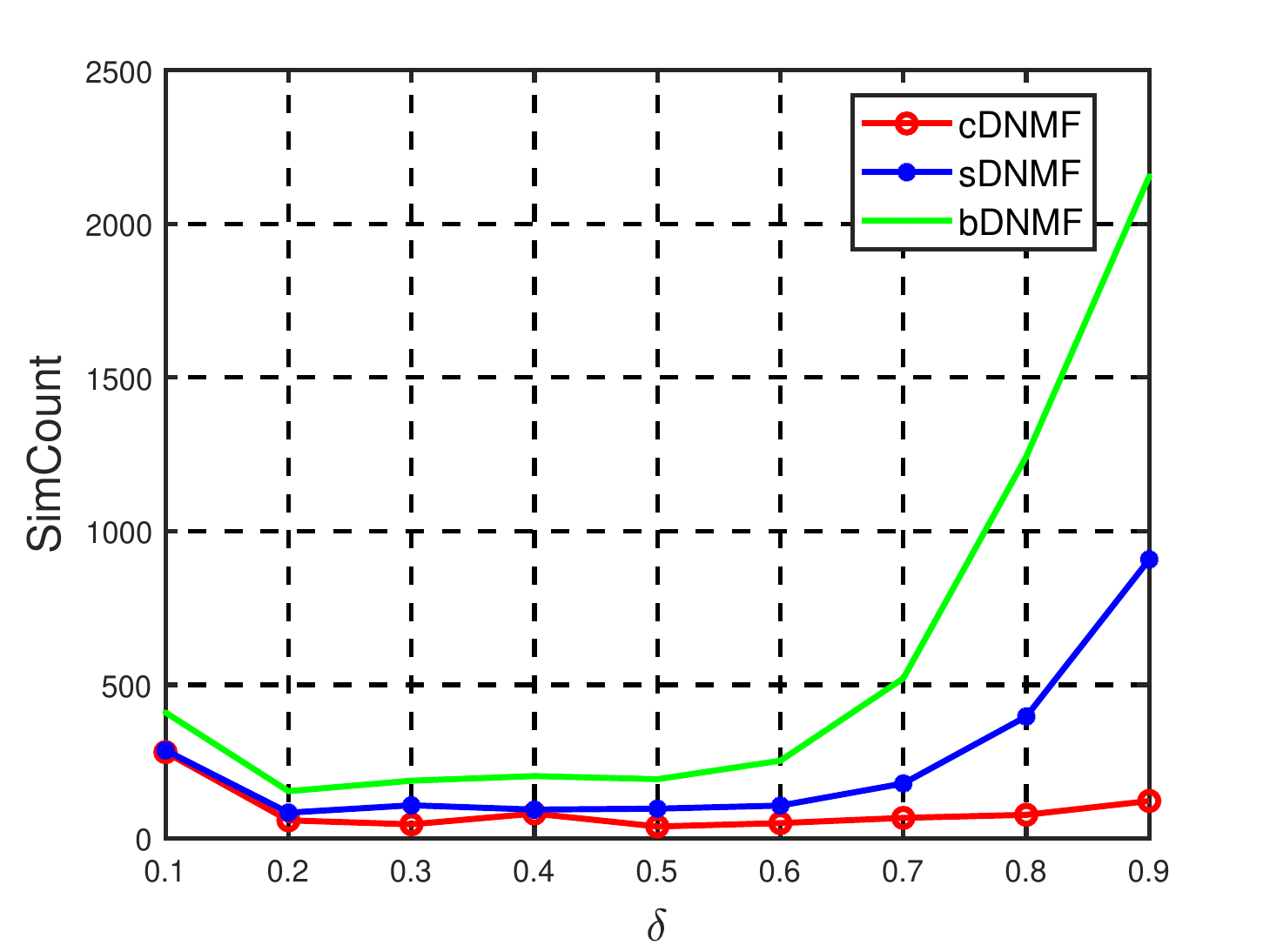}%
\label{20News_delta_SC}}
\caption{Performance of the DNMF variants with respect to hyperparameter $\delta$ on 20-newsgroups in terms of coherence and similarity count.}
\label{20News_delta}
\end{figure}

\begin{figure}[t]
\vspace{-0.7cm} %调整图片与上文的垂直距离
\setlength{\abovecaptionskip}{0cm} %缩小caption和图像之间的距离
\setlength{\belowcaptionskip}{-0cm} %缩小caption和下方文字的距离
\centering
\subfloat[]{\includegraphics[width=1.75in]{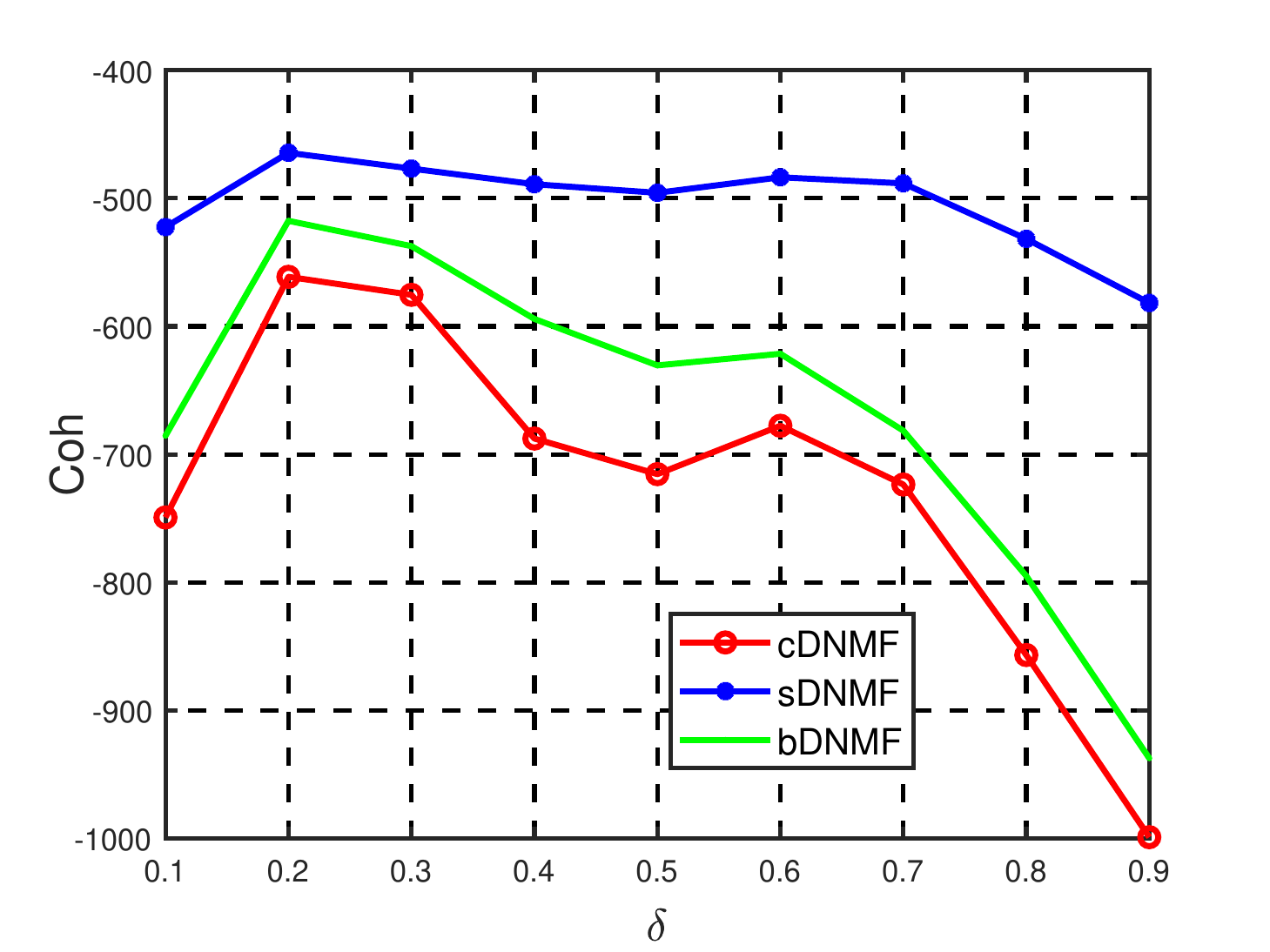}%
\label{TDT2_delta_Coh}}
\hfil
\subfloat[]{\includegraphics[width=1.75in]{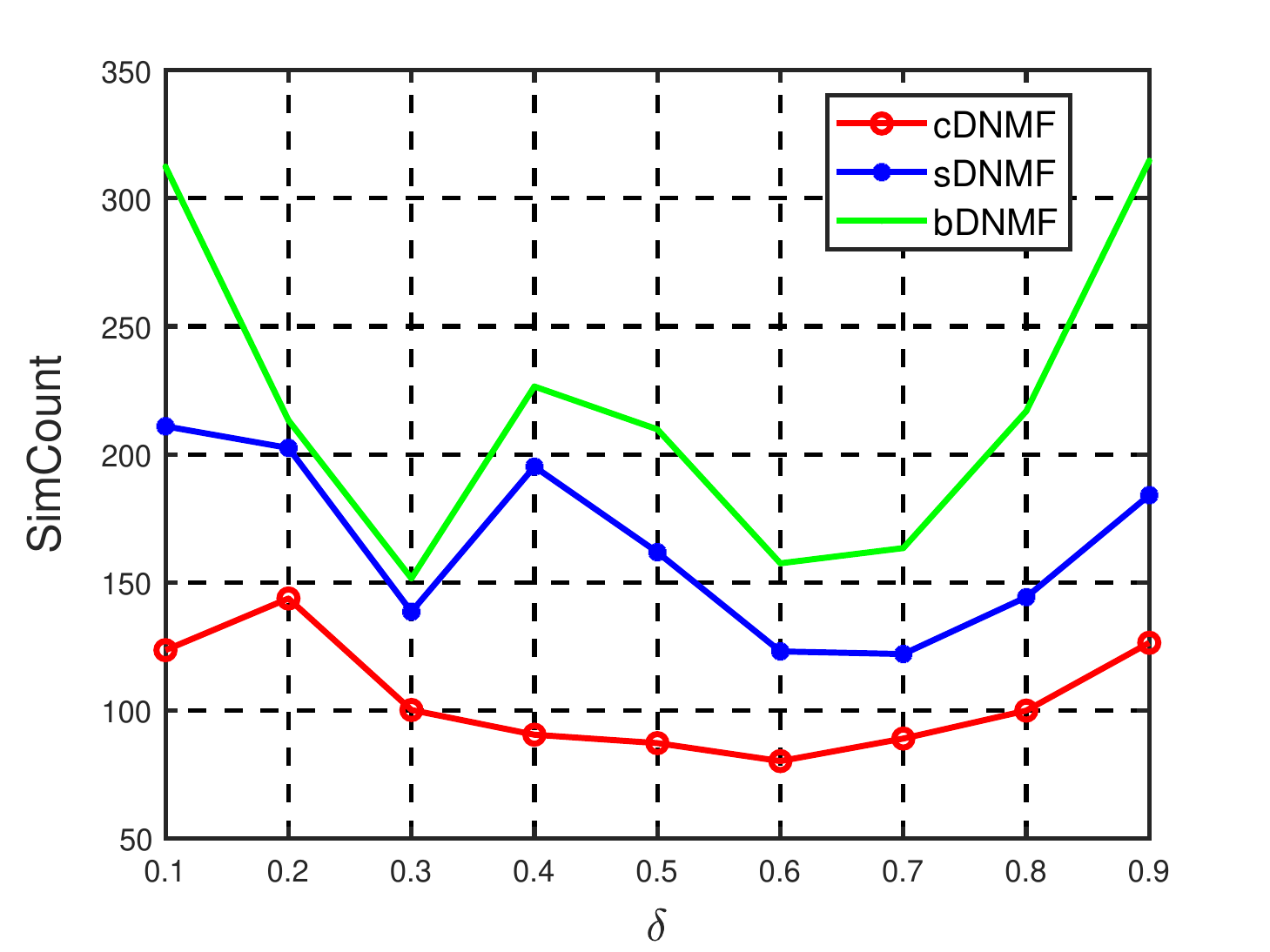}%
\label{TDT2_delta_SC}}
\caption{Performance of the DNMF variants with respect to hyperparameter $\delta$ on TDT2 in terms of coherence and similarity count.}
\label{TDT2_delta}
\end{figure}

\begin{figure}[!t]
\vspace{-0.7cm} %调整图片与上文的垂直距离
\setlength{\abovecaptionskip}{0cm} %缩小caption和图像之间的距离
\setlength{\belowcaptionskip}{-0cm} %缩小caption和下方文字的距离
\centering
\subfloat[]{\includegraphics[width=1.75in]{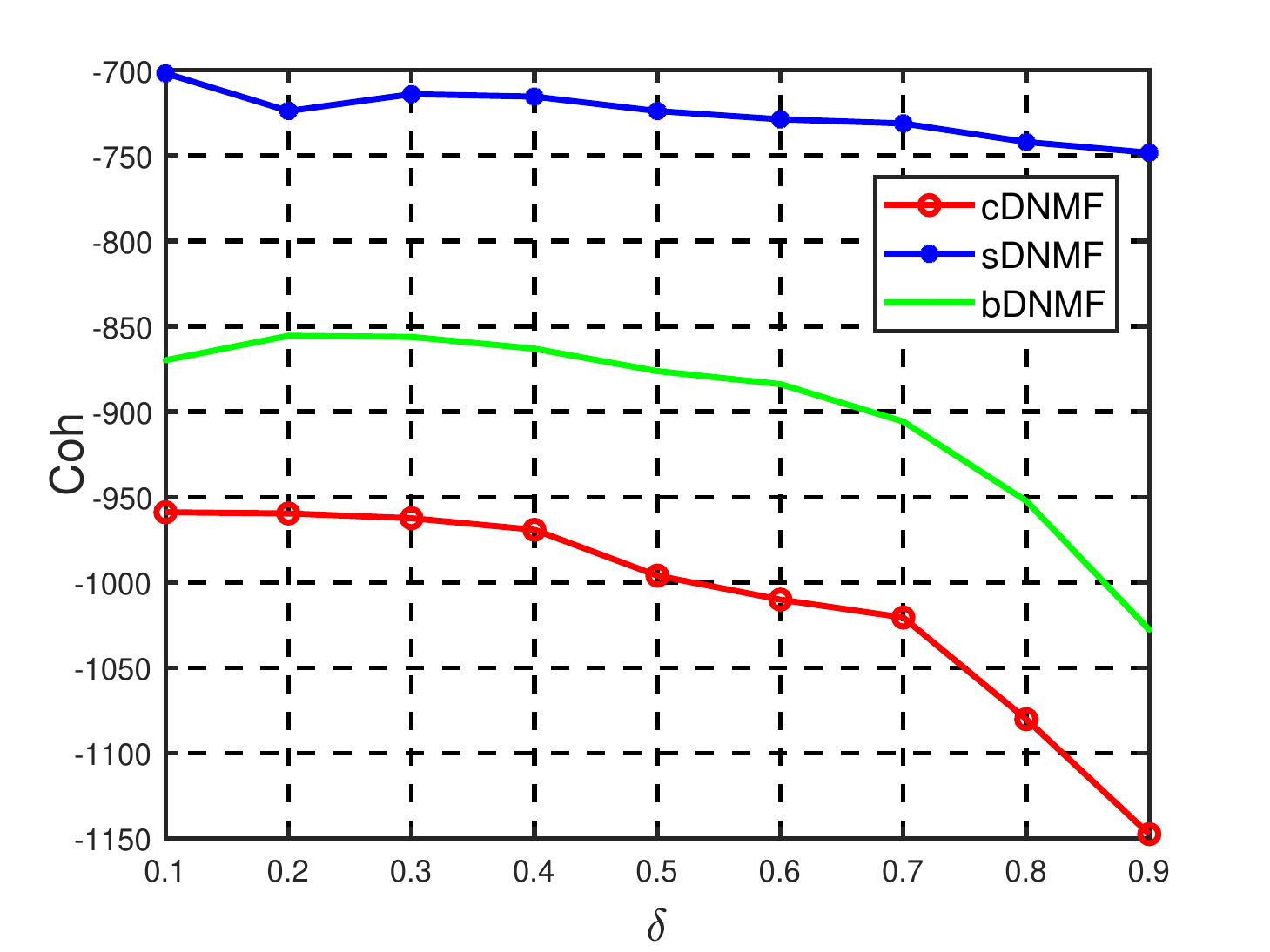}%
\label{Reuters_delta_Coh}}
\hfil
\subfloat[]{\includegraphics[width=1.75in]{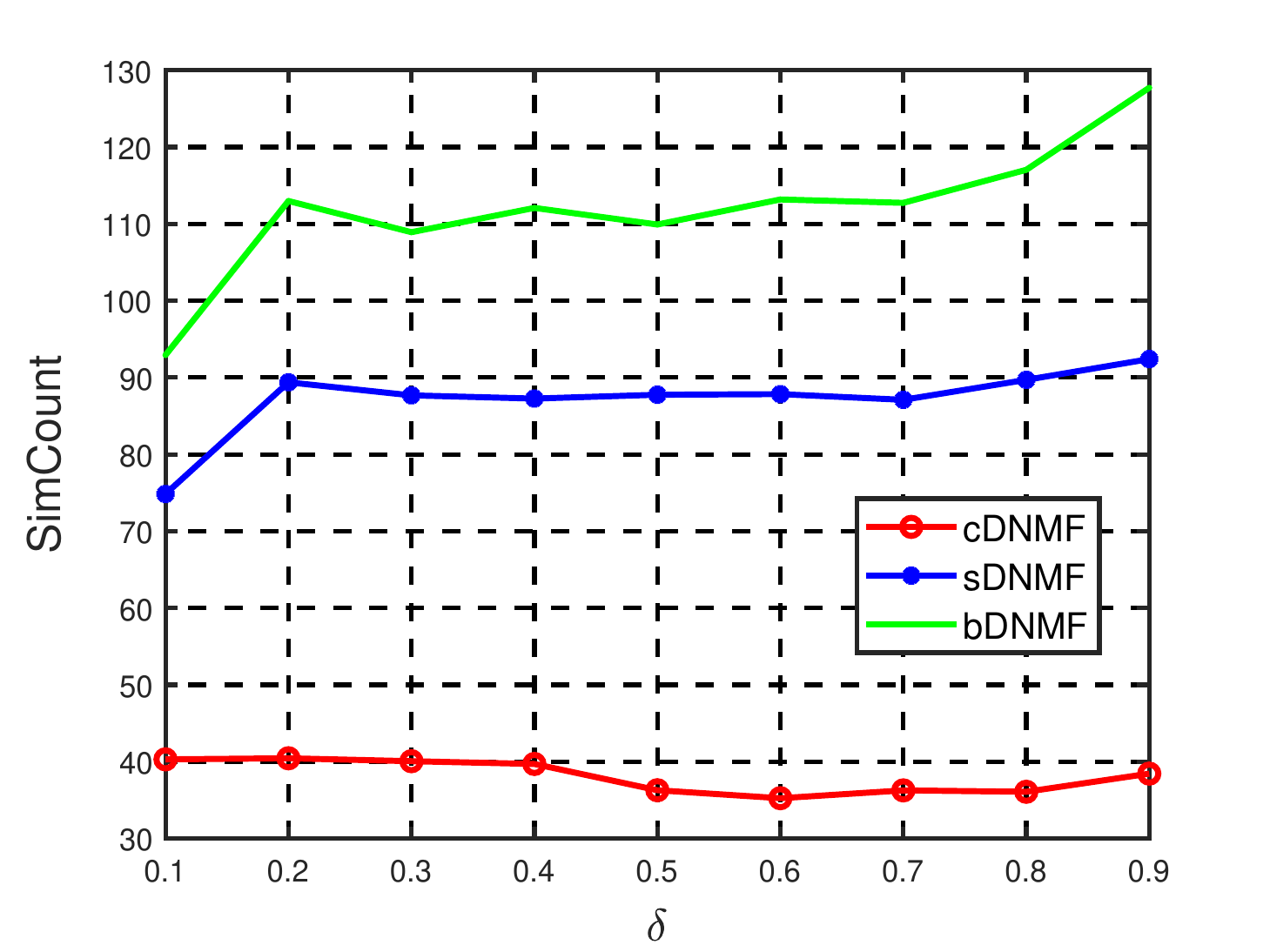}%
\label{Reuters_delta_SC}}
\caption{Performance of the DNMF variants with respect to hyperparameter $\delta$ on Reuters-21578 in terms of coherence and similarity count.}
\label{Reuters_delta}
\end{figure}

\begin{figure}[t]
\vspace{-0.7cm} %调整图片与上文的垂直距离
\setlength{\abovecaptionskip}{0cm} %缩小caption和图像之间的距离
\setlength{\belowcaptionskip}{-0cm} %缩小caption和下方文字的距离
\centering
\subfloat[]{\includegraphics[width=1.75in]{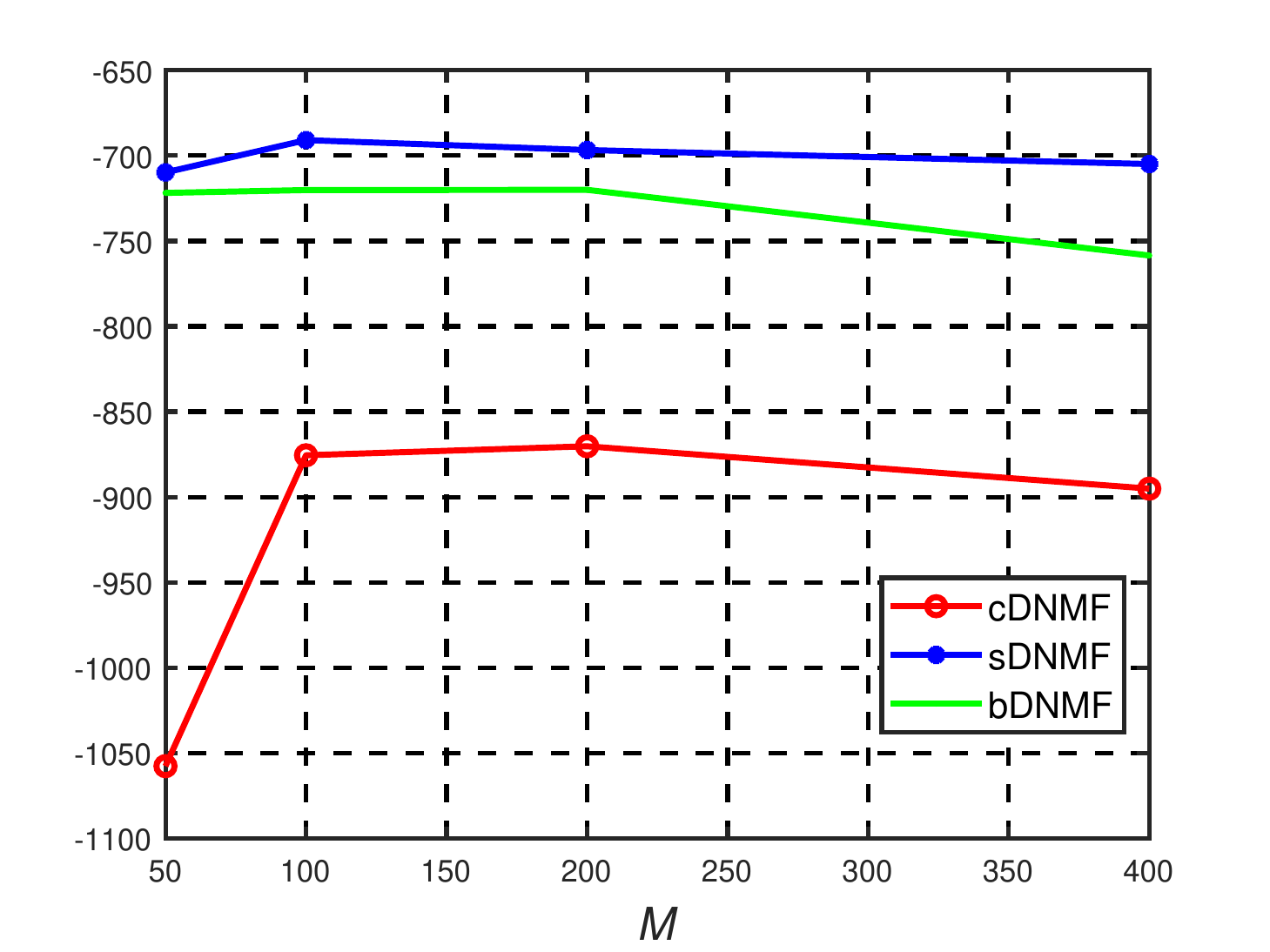}%
\label{20News_V_Coh}}
\hfil
\subfloat[]{\includegraphics[width=1.75in]{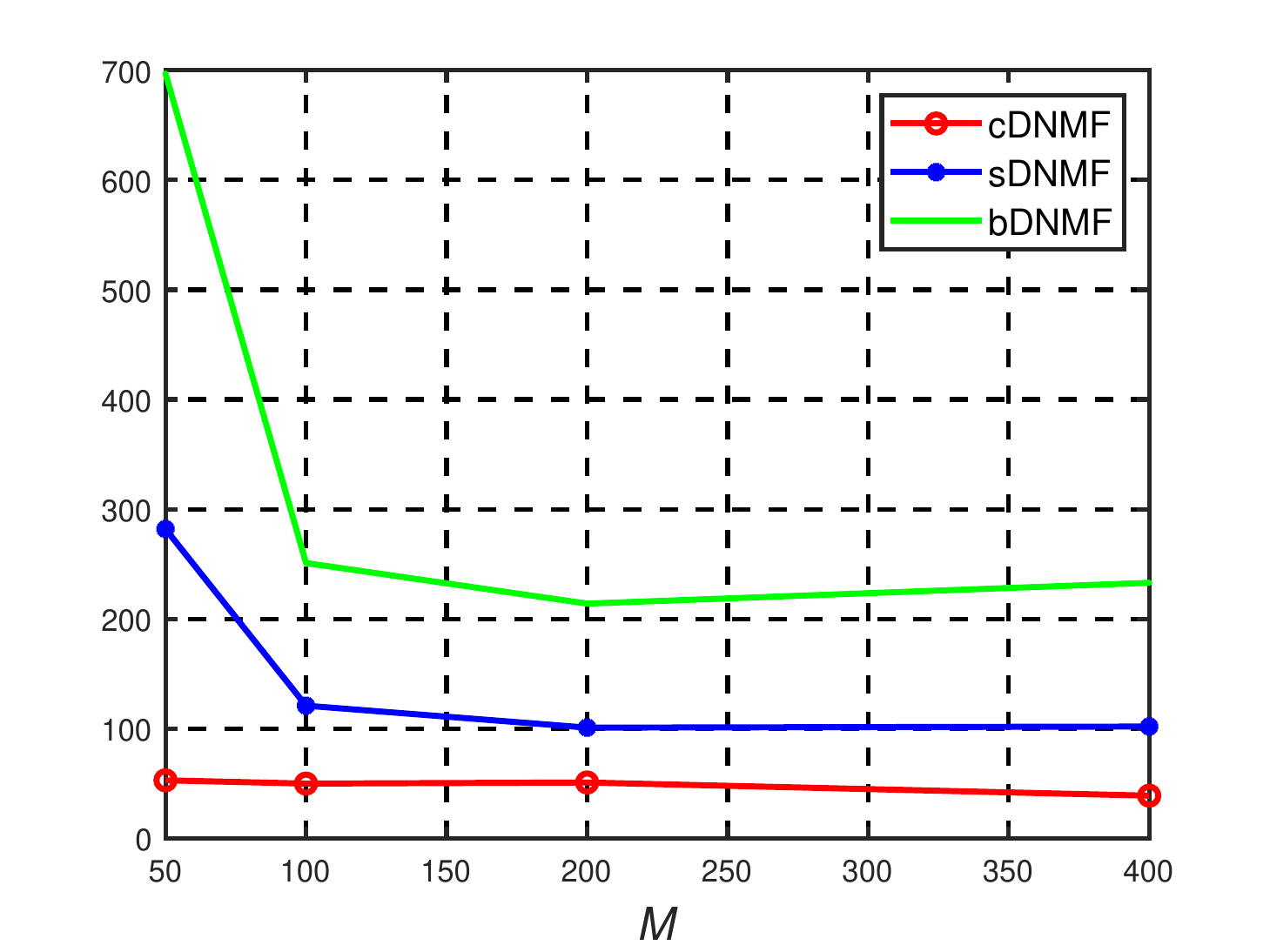}%
\label{20News_V_SC}}
\caption{Performance of the DNMF variants with respect to hyperparameter $M$ on 20-newsgroups in terms of coherence and similarity count.}
\label{20News_V}
\end{figure}

\begin{figure}[t]
\vspace{-0.7cm} %调整图片与上文的垂直距离
\setlength{\abovecaptionskip}{0cm} %缩小caption和图像之间的距离
\setlength{\belowcaptionskip}{-0cm} %缩小caption和下方文字的距离
\centering
\subfloat[]{\includegraphics[width=1.75in]{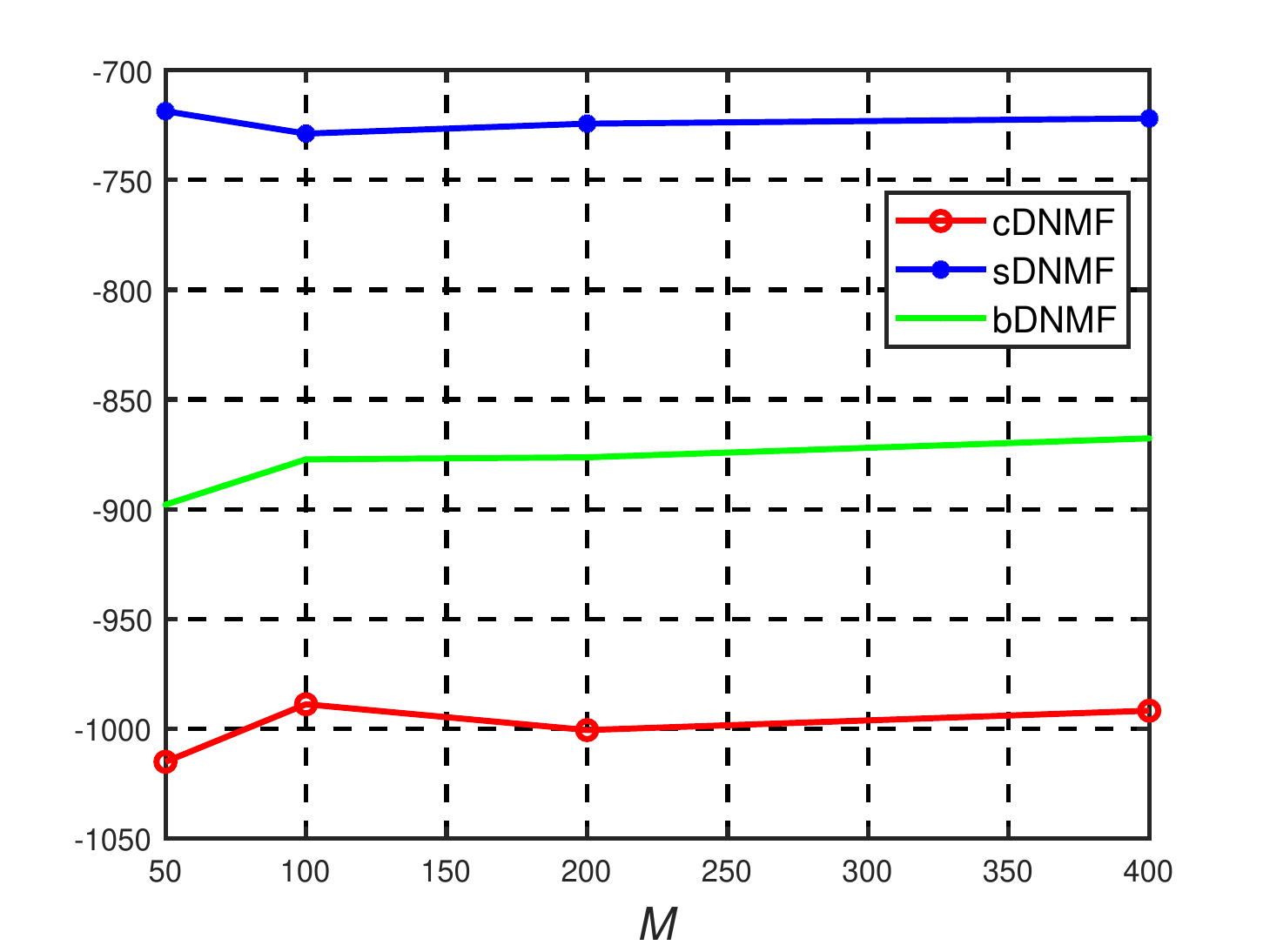}%
\label{TDT2_V_Coh}}
\hfil
\subfloat[]{\includegraphics[width=1.75in]{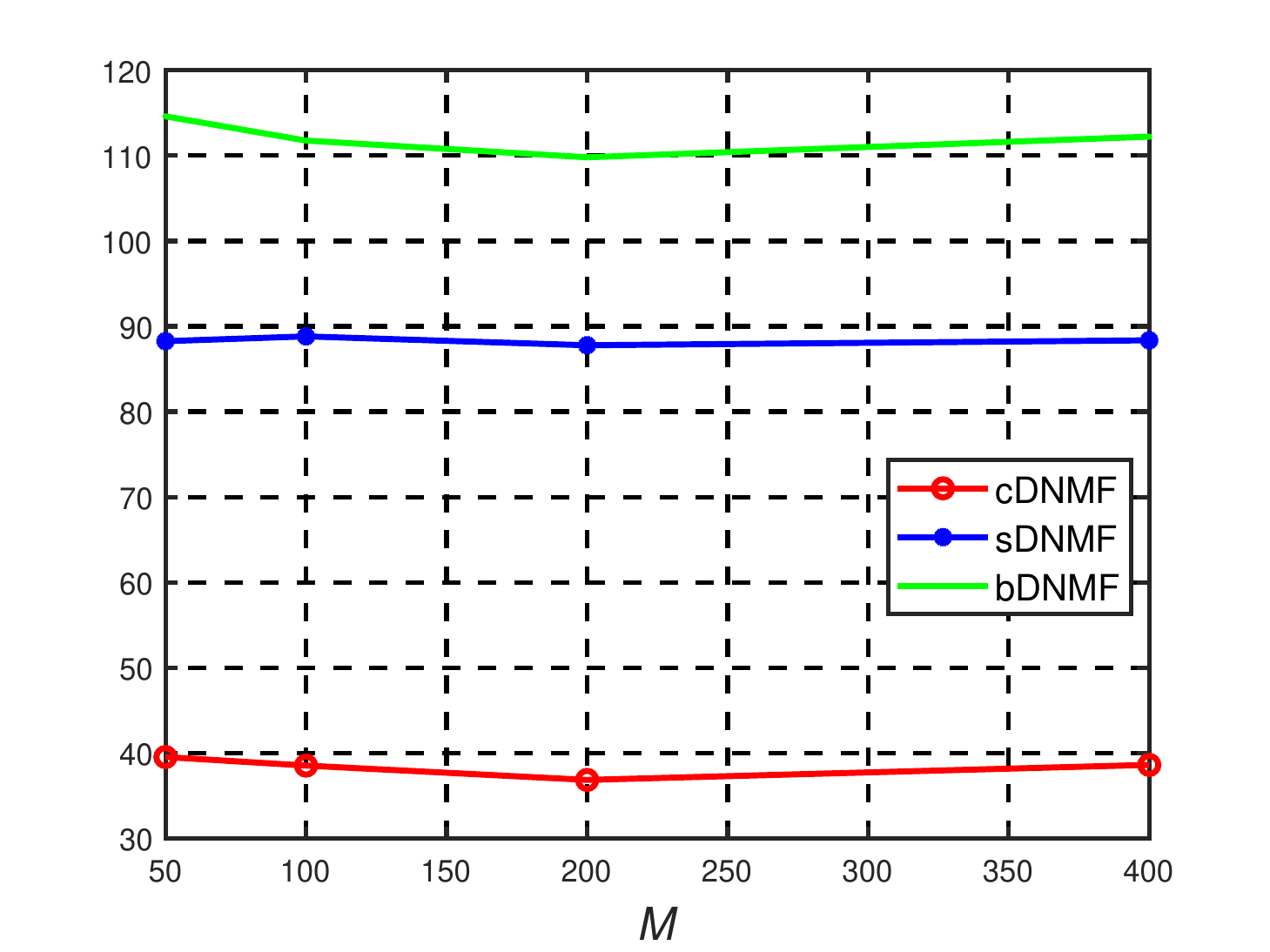}%
\label{TDT2_V_SC}}
\caption{Performance of the DNMF variants with respect to hyperparameter $M$ on TDT2 in terms of coherence and similarity count.}
\label{TDT2_V}
\end{figure}

\begin{figure}[!t]
\vspace{-0.7cm} %调整图片与上文的垂直距离
\setlength{\abovecaptionskip}{0cm} %缩小caption和图像之间的距离
\setlength{\belowcaptionskip}{-0cm} %缩小caption和下方文字的距离
\centering
\subfloat[]{\includegraphics[width=1.75in]{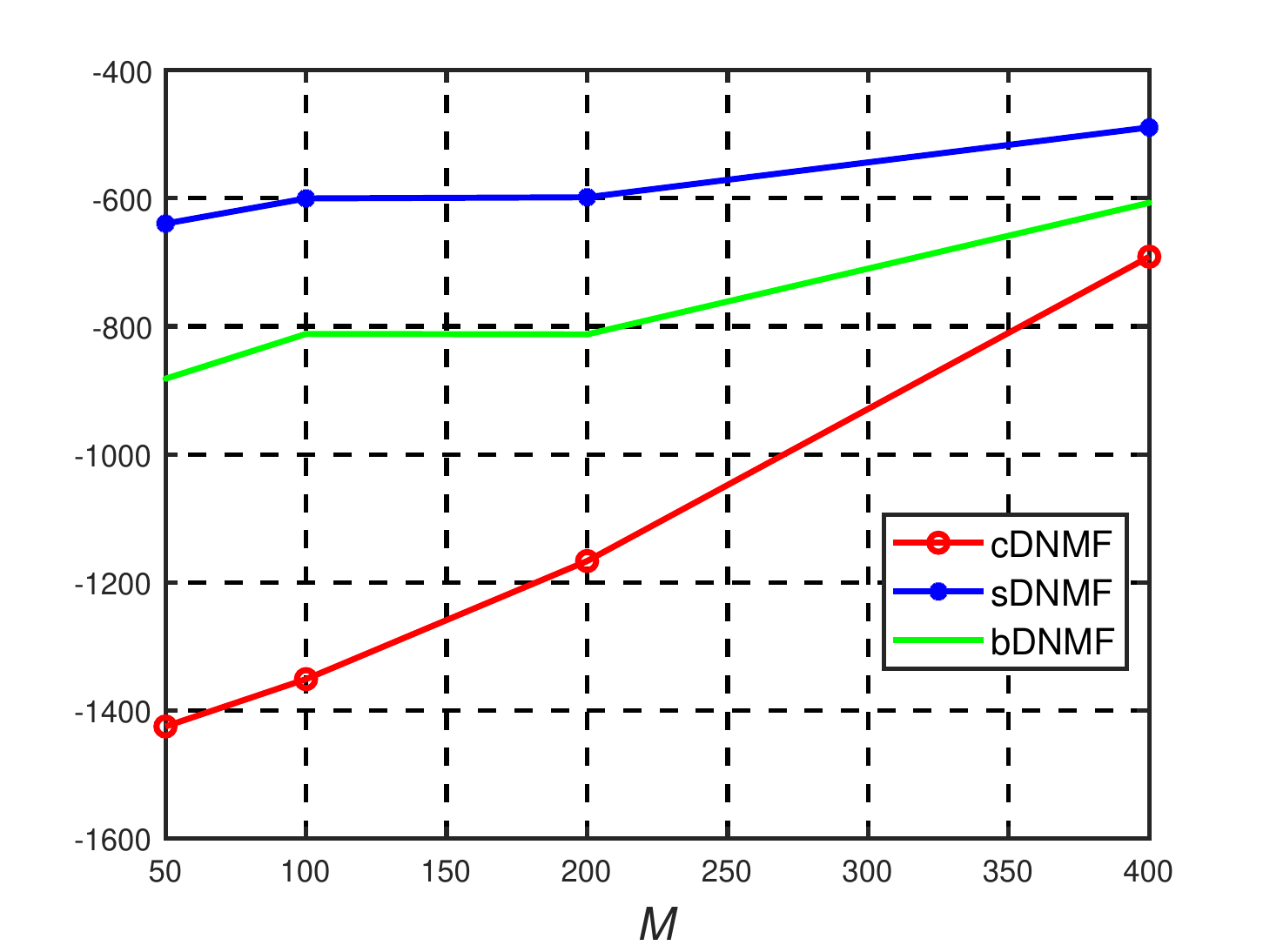}%
\label{Reuters_V_Coh}}
\hfil
\subfloat[]{\includegraphics[width=1.75in]{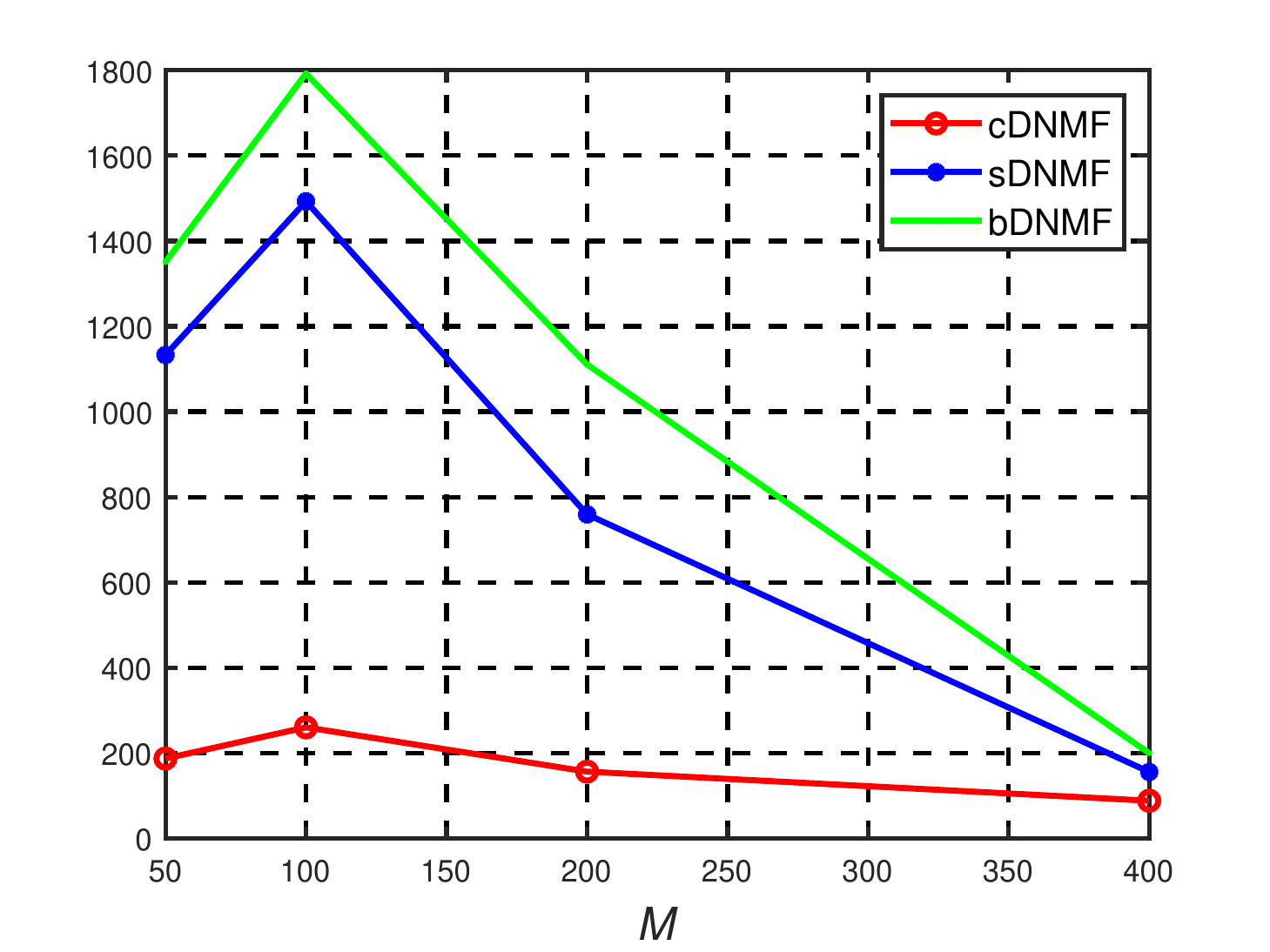}%
\label{Reuters_V_SC}}
\caption{Performance of the DNMF variants with respect to hyperparameter $V$ on Reuters-21578 in terms of coherence and similarity count.}
\label{Reuters_V}
\end{figure}

 \subsection{Effects of the hyperparameters of DNMF}
To study how the hyperparameters of DNMF affect the performance, we searched the hyperparameters in grid. To prevent exhaust search, when we studied a hyperparameter, we fixed the other hyperparameters to their default values.

We first studied the two regularization hyperparameters of cDNMF $\lambda_1$ and $\lambda_2$ in terms of coherence and similarity count by searching the two hyperparametres in grid from 0.1 to 0.9. The results are shown in Figs. \ref{lambda_1} and \ref{lambda_2}. From the figures, we see that cDNMF is in sensitive to the two hyperparameters.

Then, we studied the hyperparameters $\delta$ and $M$ of the deep model in DNMF in terms of all three evaluation metrics, in which $\delta$ is searched from 0.1 to 0.9 and $M$ searched from 10 to 400. Figure \ref{ACC} shows the clustering accuracy of DNMF with respect to the two hyperparameters. Figures \ref{20News_delta} to \ref{Reuters_delta} shows the coherence and similarity count of the DNMF variants with respect to $\delta$ on the three corpora respectively. Figures \ref{20News_V} to \ref{Reuters_V} shows the coherence and similarity count of the DNMF variants with respect to $M$ on the three corpora respectively.
From the figures, we see that although the DNMF variants are sensitive to $\delta$ and $M$, we can clearly find the regulations. For the hyperparameter $\delta$, we observe from Fig. \ref{ACC}a and Figs. \ref{20News_delta} to \ref{Reuters_delta} that, when $\delta$ is set around the default value $0.5$, all DNMF variants approach to the top performance in all cases.

For the hyperparameter $M$, we see from Fig. \ref{ACC}b that enlarging $M$ clearly improves the clustering accuracy of all DNMF variants. From Figs. \ref{20News_V}a, \ref{TDT2_V}a, and \ref{Reuters_V}a, we see that the performance of all DNNF variants is improved generally along with the increase of $M$ in terms of coherence in all cases except that the performance of bDNMF and sDNMF is getting worse on 20-newsgroups. From Figs. \ref{20News_V}b, \ref{TDT2_V}b, and \ref{Reuters_V}b, we see that the similarity count scores of all DNMF variants are getting smaller generally along with the increase of $M$ on 20-newsgroups and TDT2. It is interesting to observe that the similarity count scores of the DNMF variants first get larger and then smaller along with the increase of $M$ on Reuters-21578, with a peak at $M=100$. Nonetheless, the DNMF variants approach to the lowest similarity count scores at $M=400$ in all cases. We can imagine that, when we set $M$ larger than 400, the performance may be further improved with the expense of higher computational complexity. To balance the performance and the computational complexity, it is reasonable to set $M=400$.

\section{Conclusions}
\label{Conclusion}
In this paper, we have proposed a deep NMF topic modeling framework and evaluated its effectiveness with three implementations. To our knowledge, this is the first deep NMF topic modeling framework.
The novelty of DNMF lies in the following aspects.
First, we proposed a novel unsupervised deep NMF framework that is fundamentally different from existing deep learning based topic modeling methods. It takes the unsupervised deep learning as a constraint of NMF. It is a general framework that can incorporate many types of deep models and NMF methods. To evaluate its effectiveness, we implemented three DNMF algorithms, denoted as bDNMF, sDNMF, and cDNMF. bDNMF takes the sparse output of the deep model as the topic-document matrix directly, which formulates bDNMF as a supervised regression problem with a nonnegative constraint on the word-topic matrix. sDNMF takes the output of the deep model as a mask of the topic-document matrix, and solves the NMF problem by the alternative iterative optimization, which relaxes the strong constraint on the topic-document matrix in bDNMF. cDNMF takes the output of the deep model as a regularization, which further relaxes the constraint on the topic-document matrix. To our knowledge, the regularization terms in cDNMF is novel. Finally, we applied multilayer bootstrap networks for document clustering. It reaches the state-of-the-art performance given the high-dimensional sparse TF-IDF statistics of the documents, which further boosts the overall performance of the DNMF implementations. We have conducted an extensive experimental comparison with 9 representative comparison methods covering probabilistic topic models, NMF topic models, and deep topic modeling on three benchmark datasets---20-newsgroups, TDT2, and Reuters-21578. Experimental results show that the proposed DNMF variants outperform the comparison methods significantly in terms of clustering accuracy, coherence, and similarity count. Moreover, although the DNMF variants are relative sensitive to the hyperparameter $\delta$, we always find a robust working range across the corpora, which demonstrates the robustness of the DNMF variants in real-world applications.

%\appendices
%
%\ifCLASSOPTIONcompsoc
%\else
%  % regular IEEE prefers the singular form
%\fi
%
%
%The authors would like to thank Dr.Kejun Huang for sharing the source code of the AnchorFree algorithm.

\ifCLASSOPTIONcaptionsoff
  \newpage
\fi

\bibliographystyle{IEEEbib}
\bibliography{strings,refs}

%\end{IEEEbiography}
%% insert where needed to balance the two columns on the last page with
%%\newpage

\begin{appendices}
\section{  }\label{appendices_A}

Before we prove Theorem \ref{theorem1}, we first give the definition of an upper bound auxiliary function.
\begin{myDef}
$\mathcal{G}(u,u^\prime)$ is an upper bound auxiliary function for $g(u)$ if the following conditions are satisfied:
\begin{equation}\label{eq:auxiliary_function}
\begin{split}
\mathcal{G}(u,u^\prime) \geq g(u), \mathcal{G}(u,u) = g(u)
\end{split}
\end{equation}
\end{myDef}
\begin{myCor}
If $\mathcal{G}(\cdot,\cdot)$ is an upper bound auxiliary function for $g(u)$, then $g(u)$ is non-increasing under the update rule
\begin{equation}\label{eq:Lemma1}
\begin{split}
u^{t+1} = \arg\min_{u} \mathcal{G}(u,u^t)
\end{split}
\end{equation}
\end{myCor}

\begin{myPro}\label{Pro}
For any matrices $\mathbf{A}\in\mathbb{R}^{n\times n}_+$, $\mathbf{B}\in\mathbb{R}^{k\times k}_+$, $\mathbf{E}\in\mathbb{R}^{n\times k}_+$, $\mathbf{E}^\prime\in\mathbb{R}^{n\times k}_+$, with $\mathbf{A}$ and $\mathbf{B}$ being symmetric matrices, the following inequality holds \cite{ding2008convex}:
\begin{equation}\label{eq:Pro1}
\begin{split}
\sum_{i=1}^n\sum_{j=1}^k \frac{[\mathbf{A}\mathbf{E}^\prime\mathbf{B}]_{ij}[\mathbf{E}]_{ij}^2} {[\mathbf{E}]^\prime_{ij}} \geq \mbox{Tr}(\mathbf{E}^T\mathbf{A}\mathbf{E}\mathbf{B})
\end{split}
\end{equation}
\end{myPro}

\begin{myDef}
A function can be represented as an infinite sum of terms that are calculated from the values of the function's derivatives at a single point, which can be formulated as follows:
\begin{equation}\label{eq:Taylor}
\begin{split}
f(x) = \sum_{n=0}^\infty \frac{f^{(n)}(a)}{n!}(y-a)^n
\end{split}
\end{equation}
\end{myDef}

Given the above definitions, the objective function of cDNMF \eqref{eq:mDNMF} with respect to the three univariate functions are obtained as follows:
\begin{equation}\label{eq:Obj_func_C}
\begin{split}
\mathcal{O}_{\mathbf{C}} = \| \mathbf{D} - \mathbf{C}\W\|_F^2 + \lambda_1 \| \mathbf{C} \mathbf{C}^T - \mathbf{D} \mathbf{D}^T \|_F^2
\end{split}
\end{equation}
\begin{equation}\label{eq:Obj_func_W}
\begin{split}
\mathcal{O}_{\mathbf{W}} = \| \mathbf{D} - \mathbf{C}\W\|_F^2 + \lambda_2 \| f(\mathbf{D}) - \mathbf{T}\W \|_F^2
\end{split}
\end{equation}
\begin{equation}\label{eq:Obj_func_T}
\begin{split}
\mathcal{O}_{\mathbf{T}} = \lambda_2 \| f(\mathbf{D}) - \mathbf{T}\W \|_F^2
\end{split}
\end{equation}
Then, we have the following three lemmas.

\begin{lem}\label{lemxx}

The auxiliary function for $\mathcal{O}(\mathbf{C})$ is as follows:
\begin{equation}\label{eq:Aux_func_C}
\begin{split}
& \mathcal{G}([\mathbf{C}]_{ij},[\mathbf{C}^\prime]_{ij}) =  \mathcal{O}_\mathbf{C} + [-2\mathbf{D}\mathbf{W}^T + 2\mathbf{C}\W\W^T \\
& -4\lambda_1 \mathbf{D}\mathbf{D}^T\mathbf{C} + 4\lambda_1 \mathbf{C}\mathbf{C}^T\mathbf{C}]_{ij}([\mathbf{C}]_{ij} - [\mathbf{C}^\prime]_{ij}) \\
& + \frac{1}{3!}4\lambda_1 [\mathbf{C}]_{ij}([\mathbf{C}]_{ij} - [\mathbf{C}^\prime]_{ij})^3 + \frac{1}{4!}4\lambda_1([\mathbf{C}]_{ij} - [\mathbf{C}^\prime]_{ij})^4 \\
& + \frac{1}{2}\frac{2[\mathbf{C}\W\W^T]_{ij} + 4\lambda_1[\mathbf{C}\mathbf{C}^T\mathbf{C}]_{ij}} {[\mathbf{C}]_{ij}}([\mathbf{C}]_{ij} - [\mathbf{C}^\prime]_{ij})^2
\end{split}
\end{equation}

\begin{proof}
It is obvious that $\mathcal{G}(\mathbf{C},\mathbf{C}) = \mathcal{O}_\mathbf{C}(\mathbf{C})$, we only need to prove that $\mathcal{G}(\mathbf{C},\mathbf{C}^\prime) \geq \mathcal{O}_\mathbf{C}(\mathbf{C})$.
The first-order partial derivative of \eqref{eq:Obj_func_C} in element-wise is
\begin{equation}\label{eq:dif1_func_C}
\begin{split}
&\frac{\partial\mathcal{O}_{\mathbf{C}}}{\partial [\mathbf{C}]_{ij}} \\
= &[-2\mathbf{D}\mathbf{W}^T + 2\mathbf{C}\W\W^T -4\lambda_1 \mathbf{D}\mathbf{D}^T\mathbf{C} + 4\lambda_1 \mathbf{C}\mathbf{C}^T\mathbf{C}]_{ij}
\end{split}
\end{equation}
The second-order derivative of \eqref{eq:Obj_func_C} with respect to $\mathbf{C}$ is
\begin{equation}\label{eq:dif2_func_C}
\begin{split}
\frac{\partial^2\mathcal{O}_{\mathbf{C}}}{\partial [\mathbf{C}]_{ij}\partial [\mathbf{C}]_{ij}} = [2\W\W^T]_{jj} -4\lambda_1 [\mathbf{D}\mathbf{D}^T]_{ii} + 4\lambda_1 [\mathbf{C}\mathbf{C}^T]_{ii}
\end{split}
\end{equation}
The third-order partial derivative of \eqref{eq:Obj_func_C} is
\begin{equation}\label{eq:dif3_func_C}
\begin{split}
\frac{\partial^3\mathcal{O}_{\mathbf{C}}}{\partial [\mathbf{C}]_{ij}\partial [\mathbf{C}]_{ij}\partial [\mathbf{C}]_{ij}} = 4\lambda_1 [\mathbf{C}]_{ij}
\end{split}
\end{equation}
The fourth-order partial derivative of \eqref{eq:Obj_func_C} is
\begin{equation}\label{eq:dif4_func_C}
\begin{split}
\frac{\partial^4\mathcal{O}_{\mathbf{C}}}{\partial [\mathbf{C}]_{ij}\partial [\mathbf{C}]_{ij}\partial [\mathbf{C}]_{ij}\partial [\mathbf{C}]_{ij} } = 4\lambda_1
\end{split}
\end{equation}
According to the Taylor expansion in Definition \eqref{eq:Taylor}, we can rewrite \eqref{eq:Obj_func_C} to its Taylor expansion form:
\begin{equation}\label{eq:Refunc_C}
\begin{split}
\mathcal{O}_\mathbf{C}(c_{ij})= & \mathcal{O}_\mathbf{C} + \frac{\partial\mathcal{O}_{\mathbf{C}}}{\partial c_{ij}}(c_{ij} - [\mathbf{C}]_{ij}) \\
& + \frac{1}{2}\frac{\partial^2\mathcal{O}_{\mathbf{C}}}{\partial c_{ij}\partial c_{ij}}(c_{ij} - [\mathbf{C}]_{ij})^2 \\
& + \frac{1}{3!}\frac{\partial^3\mathcal{O}_{\mathbf{C}}}{\partial c_{ij}\partial c_{ij}\partial c_{ij}}(c_{ij} - [\mathbf{C}]_{ij})^3 \\
& + \frac{1}{4!}\frac{\partial^4\mathcal{O}_{\mathbf{C}}}{\partial c_{ij}\partial c_{ij}\partial c_{ij}\partial c_{ij}}(c_{ij} - [\mathbf{C}]_{ij})^4
\end{split}
\end{equation}
The upper bound auxiliary function for \eqref{eq:Obj_func_C} is defined as
\begin{equation}\label{eq:auxiliary_C}
\begin{split}
& \mathcal{G}([\mathbf{C}]_{ij},[\mathbf{C}^\prime]_{ij}) =  \mathcal{O}_\mathbf{C} + \frac{\partial\mathcal{O}_{\mathbf{C}}}{\partial [\mathbf{C}]_{ij}}([\mathbf{C}]_{ij} - [\mathbf{C}^\prime]_{ij}) \\
& + \frac{1}{3!}\frac{\partial^3\mathcal{O}_{\mathbf{C}}}{\partial [\mathbf{C}]_{ij}\partial [\mathbf{C}]_{ij}\partial [\mathbf{C}]_{ij}}([\mathbf{C}]_{ij} - [\mathbf{C}^\prime]_{ij})^3 \\
& + \frac{1}{4!}\frac{\partial^4\mathcal{O}_{\mathbf{C}}}{\partial [\mathbf{C}]_{ij}\partial [\mathbf{C}]_{ij}\partial [\mathbf{C}]_{ij}\partial [\mathbf{C}]_{ij}}([\mathbf{C}]_{ij} - [\mathbf{C}^\prime]_{ij})^4 \\
& + \frac{1}{2}\frac{2[\mathbf{C}\W\W^T]_{ij} + 4\lambda_1[\mathbf{C}\mathbf{C}^T\mathbf{C}]_{ij}} {[\mathbf{C}]_{ij}}([\mathbf{C}]_{ij} - [\mathbf{C}^\prime]_{ij})^2
\end{split}
\end{equation}
Substituting \eqref{eq:Refunc_C} into \eqref{eq:auxiliary_C}, we find that $\mathcal{G}(\mathbf{C},\mathbf{C}^\prime) \geq \mathcal{O}_\mathbf{C}(\mathbf{C})$ is equivalent to
\begin{equation}\label{eq:equivalent_C}
\begin{split}
& \frac{1}{2}\frac{2[\mathbf{C}\W\W^T]_{ij} + 4\lambda_1[\mathbf{C}\mathbf{C}^T\mathbf{C}]_{ij}} {[\mathbf{C}]_{ij}}([\mathbf{C}]_{ij} - [\mathbf{C}^\prime]_{ij})^2 \\
& \geq \frac{1}{2}\left([2\W\W^T]_{jj} -4\lambda_1[\mathbf{D}\mathbf{D}^T]_{ii} + 4\lambda_1[\mathbf{C}\mathbf{C}^T]_{ii}\right)([\mathbf{C}]_{ij} - [\mathbf{C}^\prime]_{ij})^2
\end{split}
\end{equation}
Because we have
\begin{equation}\label{eq:equivalent_C1}
\begin{split}
& \frac{[\mathbf{C}\W\W^T]_{ij}}{[\mathbf{C}]_{ij}} = \frac{\sum_j\left([\mathbf{C}]_{ij}\times[\W\W^T]_{jj}\right)} {[\mathbf{C}]_{ij}} \\
& \geq \frac{[\mathbf{C}]_{ij}\times[\W\W^T]_{jj}}{[\mathbf{C}]_{ij}} = [\W\W^T]_{jj}
\end{split}
\end{equation}
\begin{equation}\label{eq:equivalent_C2}
\begin{split}
& \frac{[\mathbf{C}\mathbf{C}^T\mathbf{C}]_{ij}}{[\mathbf{C}]_{ij}} = \frac{\sum_j\left([\mathbf{C}\mathbf{C}^T]_{jj}\times[\mathbf{C}]_{ij}\right)} {[\mathbf{C}]_{ij}} \\
& \geq \frac{[\mathbf{C}\mathbf{C}^T]_{ii}\times[\mathbf{C}]_{ij}}{[\mathbf{C}]_{ij}} = [\mathbf{C}\mathbf{C}^T]_{ii}
\end{split}
\end{equation}
we can conclude that \eqref{eq:equivalent_C} holds, and \eqref{eq:auxiliary_C} is an upper bound auxiliary function for \eqref{eq:Obj_func_C}.
 Because the elements of matrix $\mathbf{C}$ is nonnegative, the third and fourth order partial derivatives are larger than zero and \eqref{eq:auxiliary_C} is a convex function, its minimum value can be achieved at
\begin{equation}\label{eq:equivalent_C2}
\begin{split}
&[\mathbf{C}^\prime]_{ij}\\
 & = [\mathbf{C}]_{ij} - \frac{[-2\mathbf{D}\mathbf{W}^T + 2\mathbf{C}\W\W^T -4\lambda_1\mathbf{D}\mathbf{D}^T\mathbf{C} + 4\lambda_1\mathbf{C}\mathbf{C}^T\mathbf{C}]_{ij}} {2\times\frac{1}{2}\frac{[2\mathbf{C}\W\W^T]_{ij} + 4\lambda_1[\mathbf{C}\mathbf{C}^T\mathbf{C}]_{ij}} {[\mathbf{C}]_{ij}}([\mathbf{C}]_{ij} - [\mathbf{C}^\prime]_{ij})^2} \\
& = [\mathbf{C}]_{ij}\frac{[\mathbf{D}\mathbf{W}^T]_{ij} + 2\lambda_1[\mathbf{D}\mathbf{D}^T\mathbf{C}]_{ij}}{[\mathbf{C}\W\W^T]_{ij} + 2\lambda_1[\mathbf{C}\mathbf{C}^T\mathbf{C}]_{ij}}
\end{split}
\end{equation}
Lemma \ref{lemxx} is proved.
\end{proof}

\end{lem}
\begin{lem}
Given Proposition \ref{Pro}, the auxiliary function for $\mathcal{O}(\mathbf{W})$ is as follows:
\begin{equation}\label{eq:Aux_func_W}
\begin{split}
\mathcal{G}(\mathbf{W},\mathbf{W}^\prime) & = -2\lambda_2\mbox{Tr}(f(\mathbf{D})\W^T\mathbf{T}^T) - 2\mbox{Tr}(\mathbf{D}\mathbf{W}^T\mathbf{C}^T) \\
& + \sum_{ij}\frac{[\mathbf{C}^T\mathbf{C}\W]_{ij}{[\W^\prime]}^2_{ij}} {[\W]_{ij}} + \lambda_2 \sum_{ij}\frac{[\mathbf{T}^T\mathbf{T}\W]_{ij}{[\W^\prime]}^2_{ij}}{[\W]_{ij}}
\end{split}
\end{equation}
\end{lem}
\begin{lem}
The auxiliary function for $\mathcal{O}(\mathbf{T})$ is as follows:
\begin{equation}\label{eq:Aux_func_T}
\begin{split}
\mathcal{G}(\mathbf{T},\mathbf{T}^\prime) = -2\lambda_2\mbox{Tr}(f(\mathbf{D})\W^T\mathbf{T}^T) + \lambda_2 \sum_{ij}\frac{[\mathbf{T}\W\W^T]_{ij}{[\mathbf{T}^\prime]}^2_{ij}}{[\mathbf{T}]_{ij}}
\end{split}
\end{equation}
\end{lem}

With the above lemmas, we derive the update rules for each variable by minimizing their corresponding auxiliary functions:
\begin{equation}\label{eq:difAux_func_C}
\begin{split}
& \frac{\partial\mathcal{G}(\mathbf{C},\mathbf{C}^\prime)}{\partial [\mathbf{C}^\prime]_{ij}} = -2[\mathbf{D}\mathbf{W}^T]_{ij} + 2\frac{[\mathbf{C}\mathbf{W}\mathbf{W}^T]_{ij}[\mathbf{C}^\prime]_{ij}}{[\mathbf{C}]_{ij}} \\
& - 4\lambda_1[\mathbf{D}\mathbf{D}^T\mathbf{C}]_{ij} + 4\lambda_1\frac{[\mathbf{C}\mathbf{C}^T\mathbf{C}]_{ij}{[\mathbf{C}^\prime]}_{ij}}{[\mathbf{C}]_{ij}} = 0
\end{split}
\end{equation}
\begin{equation}\label{eq:difAux_func_W}
\begin{split}
& \frac{\partial\mathcal{G}(\mathbf{W},\mathbf{W}^\prime)}{\partial [\mathbf{W}^\prime]_{ij}} = -2\lambda_2[\mathbf{T}^Tf(\mathbf{D})]_{ij} - 2[\mathbf{C}^T\mathbf{D}]_{ij} \\
& + 2\frac{[\mathbf{C}^T\mathbf{C}\W]_{ij}{[\mathbf{W}^\prime]}_{ij}} {[\mathbf{W}]_{ij}} + 2\lambda_2\frac{[\mathbf{T}^T\mathbf{T}\W]_{ij}{[\mathbf{W}^\prime]}_{ij}}{[\mathbf{W}]_{ij}} = 0
\end{split}
\end{equation}
\begin{equation}\label{eq:difAux_func_T}
\begin{split}
& \frac{\partial\mathcal{G}(\mathbf{T},\mathbf{T}^\prime)}{\partial [\mathbf{T}^\prime]_{ij}} = -2[f(\mathbf{D})\W^T]_{ij} + 2\frac{[\mathbf{T}\W\W^T]_{ij}{[\mathbf{T}^\prime]}_{ij}}{[\mathbf{T}]_{ij}} = 0
\end{split}
\end{equation}
which derives
\begin{equation}\label{eq:func_C}
\begin{split}
[\mathbf{C}^\prime]_{ij} = [\mathbf{C}]_{ij}\frac{[\mathbf{D}\mathbf{W}^T]_{ij} + 2\lambda_1[\mathbf{D}\mathbf{D}^T\mathbf{C}]_{ij}}{[\mathbf{C}\W\W^T]_{ij} + 2\lambda_1[\mathbf{C}\mathbf{C}^T\mathbf{C}]_{ij}}
\end{split}
\end{equation}
\begin{equation}\label{eq:func_W}
\begin{split}
[\mathbf{W}^\prime]_{ij} = [\mathbf{W}]_{ij}\frac{\lambda_2[\mathbf{T}^Tf(\mathbf{D})]_{ij} + [\mathbf{C}^T\mathbf{D}]_{ij}}{[\mathbf{C}^T\mathbf{C}\mathbf{W}]_{ij} + \lambda_2[\mathbf{T}^T\mathbf{T}\W]_{ij}}
\end{split}
\end{equation}
\begin{equation}\label{eq:func_T}
\begin{split}
[\mathbf{T}^\prime]_{ij} = [\mathbf{T}]_{ij}\frac{[f(\mathbf{D})\W^T]_{ij}}{[\mathbf{T}\W\W^T]_{ij}}
\end{split}
\end{equation}

It can be proved that the three update rules \eqref{eq:MU_SED_C}, \eqref{eq:MU_SED_W} and \eqref{eq:MU_SED_T} are equivalent to \eqref{eq:func_C}, \eqref{eq:func_W} and \eqref{eq:func_T}, respectively.
Because the objective function in \eqref{eq:mDNMF} is lower bounded by 0, the modified DNMF converges to a stationary point. Theorem \ref{theorem1} is proved.

\end{appendices}

\end{document}